\documentclass[a4paper,USenglish]{lipics-v2018}
\nolinenumbers\hideLIPIcs

\renewcommand\path[1]{{\rmfamily\footnotesize\detokenize{#1}}}

\usepackage{microtype}
\usepackage{booktabs}

\usepackage{empheq}
\usepackage{tcolorbox}
\definecolor{blueish}{rgb}{0.122, 0.435, 0.698}
\definecolor{dagstuhlyellow}{rgb}{0.99,0.78,0.07}
\newtcbox{\colbox}{
  nobeforeafter,
  colframe=white,
  colback=dagstuhlyellow!10!white,
  arc=10pt,
  tcbox raise base}

\usepackage{xspace}

\newcommand{\cF}{\mathcal{F}}
\newcommand{\cT}{\mathcal{T}}
\newcommand{\cP}{\mathcal{P}}

\newcommand{\cC}{\mathcal{C}}

\newcommand{\op}[1]{\ensuremath{\operatorname{#1}}}

\newcommand{\tw}{\op{tw}}

\newcommand{\tr}{\op{tr}}

\newcommand{\N}{\mathbf{N}}

\newcommand{\R}{\mathbf{R}}

\theoremstyle{plain}
\newtheorem{claim}[theorem]{Claim}
\newtheorem{prop}[theorem]{Proposition}

\DeclarePairedDelimiter\paren{\lparen}{\rparen}
\DeclarePairedDelimiter\abs{\lvert}{\rvert}
\DeclarePairedDelimiter\set{\{}{\}}
\DeclarePairedDelimiterX\setc[2]{\{}{\}}{\,#1 \delimsize: #2\,}
\DeclarePairedDelimiter\mset{\{\!\!\{}{\}\!\!\}}
\DeclarePairedDelimiterX\msetc[2]{\{\!\!\{}{\}\!\!\}}{\,#1 \delimsize: #2\,}

\usepackage{tikz}
\usetikzlibrary{calc}
\tikzstyle{bullet}=[circle,draw,fill,inner sep=2pt]

\newcommand\vect[1]{\operatorname{\uppercase{\mathsf{#1}}}}
\newcommand\matr[1]{\mathsf{#1}}
\newcommand\surj    {\matr{Surj}}
\newcommand\crr     {\matr{Cr}}
\newcommand\veccrr  {\vect{CR}}
\newcommand\wl      {\matr{WL}}
\renewcommand\hom   {\matr{Hom}}
\newcommand\vechom  {\vect{Hom}}
\newcommand\vecinj  {\vect{INJ}}
\newcommand\vecind  {\vect{IND}}
\newcommand\strhom  {\matr{StrHom}}

\newcommand\ext     {\matr{Ext}}
\newcommand\inj     {\matr{Inj}}
\newcommand\dhom    {\overrightarrow{\hom}}

\newcommand\dsurj   {\overrightarrow{\surj}}
\newcommand\dsub    {\overrightarrow{\sub}}
\newcommand\sub     {\matr{Sub}}
\newcommand\ind     {\matr{Ind}}
\newcommand\aut     {\matr{Aut}}

\newcommand\bext    {\matr{bExt}}
\newcommand\biso    {\matr{bIso}}
\newcommand\dbiso   {\biso\dhom}
\newcommand\dbsurj  {\biso\dsurj}
\newcommand\dbsub   {\biso\dsub}

\newcommand\bstrhom {\matr{bStrHom}}
\newcommand\binj    {\matr{bInj}}

\newcommand{\bisocond}[2]{\biso\left( #1 \,\vert\, #2 \right)}

\newcommand{\trans}[1]{#1^{T}\xspace}

\newcommand{\perm}[3]{\begin{smallmatrix}
                      u_{#2} \dots u_{#3} \\
                      {#1}_{#2} \dots {#1}_{#3}
                     \end{smallmatrix}\xspace}

\newcommand{\Liso}[1][]{{{\mathsf L}_{\textup{iso}}^{#1}}}
\newcommand{\Fiso}[1][]{{{\mathsf F}_{\textup{iso}}^{#1}}}
\newcommand{\atp}{\operatorname{atp}}

\newsavebox{\fminibox}
\newlength{\fminilength}

\hypersetup{
  pdftitle = {Lovász Meets Weisfeiler and Leman},
  pdfauthor = {Holger Dell, Martin Grohe, Gaurav Rattan},
}

\title{Lovász Meets Weisfeiler and Leman}

\author{Holger Dell}{Saarland University and Cluster of Excellence (MMCI), Saarbrücken, Germany}{hdell@mmci.uni-saarland.de}{https://orcid.org/0000-0001-8955-0786}{}

\author{Martin Grohe}{RWTH Aachen University, Aachen, Germany}{grohe@informatik.rwth-aachen.de}{https://orcid.org/0000-0002-0292-9142}{}

\author{Gaurav Rattan}{RWTH Aachen University, Aachen, Germany}{grohe@informatik.rwth-aachen.de}{https://orcid.org/0000-0002-5095-860X}{}

\authorrunning{H. Dell, M. Grohe, and G. Rattan}

\Copyright{Holger Dell, Martin Grohe, and Gaurav Rattan}

\subjclass{
;
}


\category{}

\relatedversion{Proceedings version to appear at the 45th International Colloquium on Automata, Languages, and Programming (ICALP 2018).}

\supplement{}

\funding{}

\begin{document}

\maketitle

\begin{abstract}
  In this paper, we relate a beautiful theory by Lovász with a popular
  heuristic algorithm for the graph isomorphism problem, namely the
  color refinement algorithm and its $k$-dimensional generalization
  known as the Weisfeiler-Leman algorithm. We prove that two graphs
  $G$ and $H$ are indistinguishable by the color refinement algorithm
  if and only if, for all trees $T$, the number $\mathsf{Hom}(T,G)$ of
  homomorphisms from $T$ to $G$ equals the corresponding
  number $\mathsf{Hom}(T,H)$ for $H$. 

  There is a natural system of linear equations whose nonnegative integer solutions correspond to the isomorphisms between two graphs.
  The nonnegative real solutions to this system are called fractional isomorphisms, and two graphs are fractionally isomorphic if and only if the color refinement algorithm cannot distinguish them (Tinhofer 1986, 1991).
  We show that, if we drop the nonnegativity constraints, that is, if we look for arbitrary real solutions, then a solution to the linear system exists if and only if, for all $t$, the two graphs have the same number of length-$t$ walks.

  We lift the results for trees to an equivalence between numbers of
  homomorphisms from graphs of tree width $k$, the $k$-dimensional
  Weisfeiler-Leman algorithm, and the level-$k$ Sherali-Adams
  relaxation of our linear program. We also obtain a partial result
  for graphs of bounded path width and solutions to our system where we
  drop the nonnegativity constraints.

  A consequence of our results is a quasi-linear time algorithm to decide
  whether, for two given graphs $G$ and $H$, there is a tree $T$ with
  $\mathsf{Hom}(T,G)\neq\mathsf{Hom}(T,H)$.
 \end{abstract}

\section{Introduction}

An old result due to Lov\'asz \cite{lov67} states a graph $G$ can
be characterized by counting homomorphisms from all graphs $F$ to
$G$. That is, two graphs $G$ and $H$ are isomorphic if and only if, for all
$F$, the number $\hom(F,G)$ of homomorphisms from $F$ to $G$ equals
the number $\hom(F,H)$ of homomorphism from $F$ to $H$. This simple
result has far reaching consequences, because mapping graphs $G$ to
their \emph{homomorphism vectors}
$\vechom(G):=\big(\hom(F,G)\big)_{F\text{ graph}}$ (or suitably scaled
versions of these infinite vectors) allows us to apply
tools from functional analysis in graph theory. This is the foundation
of the beautiful theory of graph limits, developed by Lov\'asz and
others over the last 15 years (see \cite{lov12}). 

However, from a computational perspective, representing graphs by
their homomorphism vectors has the disadvantage that the problem of
computing the entries of these vectors is~NP-complete. To avoid this
difficulty, we may want to restrict the homomorphism vectors to
entries from a class of graphs for which counting homomorphisms is
tractable. That is, instead of considering the full homomorphism
vector $\vechom(G)$ we consider the vector
$\vechom_{\cF}(G):=\big(\hom(F,G)\big)_{F\in\cF}$ for a class $\cF$ of
graphs such that the problem of computing $\hom(F,G)$ for given graphs
$F\in\cF$ and $G$ is in polynomial time. Arguably the most natural
example of such a class $\cF$ is the class of all trees. More
generally, computing $\hom(F,G)$ for given graphs $F\in\cF$ and $G$ is
in polynomial time for all classes $\cF$ of bounded tree width, and
under a natural assumption from parameterized complexity theory, it is
not in polynomial time for any class $\cF$ of unbounded tree width
\cite{daljon04}. This immediately raises the question what the vector
$\vechom_{\cF}(G)$, for a class $\cF$ of bounded tree width, tells us
about the graph $G$. 

A first nice example (Proposition~\ref{prop:spectrum}) is that the
vector $\vechom_{\cC}(G)$ for the class $\cC$ of all cycles
characterizes the spectrum of a graph, that is, for graphs $G,H$ we
have $\vechom_{\cC}(G)=\vechom_{\cC}(H)$ if and only if the
adjacency matrices of $G$ and $H$ have the same eigenvalues with the
same multiplicities. This equivalence is a basic observation in spectral graph theory (see~\cite[Lemma~1]{van2003graphs}).
Before we state deeper results along these lines, let us describe a different (though related) motivation for this research. 

Determining the similarity between two graphs is an important
problem with many applications, mainly in machine learning, where it
is known as ``graph matching''~(e.g.~\cite{confogsanven04}).
But how can the similarity between graphs be measured?
An obvious idea is to use the \emph{edit distance}, which simply counts how many edges and vertices have to be deleted from or added to one graph to obtain the other.
However, two graphs that have a small edit distance can nevertheless be structurally quite dissimilar (e.g.~\cite[Section 1.5.1]{lov12}).
The edit distance is also very hard to compute as it is closely related to the
notoriously difficult quadratic assignment problem~{(e.g.~\cite{arvkobkuhvas12,nagsvi09})}.

Homomorphism vectors offer an alternative, more structurally oriented approach to measuring graph similarity.
After suitably scaling the vectors, we can can
compare them using standard vector norms. This idea is reminiscent of the ``graph kernels'' used in machine learning~{(e.g.\
  \cite{visschrakonbor10})}. Like the homomorphism vectors, many graph
kernels are based on the idea of counting certain patterns in graphs,
such as paths, walks, cycles or subtrees, and in fact any inner
product on the homomorphism vectors yields a graph kernel.

A slightly different type of graph kernel is the so-called Weisfeiler-Leman
(subtree) kernel~\cite{sheschlee+11}.
This kernel is derived from the \emph{color refinement} algorithm (a.k.a.\ the \emph{1-dimensional Weisfeiler-Leman algorithm}), which is a simple and efficient heuristic to test whether two graphs are isomorphic (e.g.\ \cite{grokermlaschwe17+}).
The algorithm computes a coloring of the vertices of a graph based on the iterated degree sequences, we give the details in Section~\ref{sec:tree}.
To use it as an isomorphism test, we
compare the color patterns of two graphs. If they are
different, we say that color refinement \emph{distinguishes} the
graphs. If the color patterns of the two graphs turn out to be the same, the graphs may still be non-isomorphic, but the algorithm fails to detect this.

Whether color refinement is able to distinguish two graphs~$G$ and~$H$ has a very nice linear-algebraic characterization due to Tinhofer~\cite{tin86,tin91}.
Let~$V$ and $W$ be the vertex sets and let $A\in\{0,1\}^{V\times V}$ and $B\in\{0,1\}^{W\times W}$ be the adjacency matrices of~$G$ and $H$, respectively.
Now consider the system~$\Fiso(G,H)$ of linear equations:
\begin{center}
  \vspace{-1em}
  \begin{minipage}{7cm}
    \begin{empheq}[
      left={\Fiso(G,H):\quad\empheqlbrace},
      box=\colbox
      ]{align}
      AX  &=XB \tag{F1}\label{eq:F1}\\
      X\boldsymbol 1_W    &=\boldsymbol 1_V    \tag{F2}\label{eq:F2}\\
      \boldsymbol 1_V^TX  &=\boldsymbol 1_W^T  \tag{F3}\label{eq:F3}
    \end{empheq}
  \end{minipage}
\end{center}

In these equations, $X$ denotes a $(V\times W)$-matrix of variables and $\boldsymbol 1_{U}$ denotes the all-1 vector over the index set~$U$.
Equations~\eqref{eq:F2} and~\eqref{eq:F3} simply state that all row and column sums
of~$X$ are supposed to be $1$.
Thus the nonnegative integer solutions to $\Fiso(G,H)$ are permutation
matrices, which due to~\eqref{eq:F1} describe isomorphisms between~$G$ and~$H$.
The nonnegative real solutions to $\Fiso(G,H)$, which in fact are always rational, are called \emph{fractional isomorphisms} between~$G$ and~$H$.
Tinhofer proved that two graphs are fractionally isomorphic if and only if color refinement does not distinguish them.

For every $k\ge 2$, color refinement has a generalization, known as
the \emph{$k$-dimensional Weisfeiler-Leman algorithm ($k$-WL)}, which colors
not the vertices of the given graph but $k$-tuples of vertices.
Atserias and Maneva~\cite{atsman13} (also see~\cite{mal14}) generalized
Tinhofer's theorem by establishing a close correspondence between $k$-WL and
the level-$k$ Sherali-Adams relaxation of $\Fiso(G,H)$.

\subsection*{Our results}

How expressive are homomorphism vectors $\vechom_{\cF}(G)$ for restricted graph classes $\cF$\,?
We consider the class $\cT$ of trees first, where the answer is surprisingly clean.

\begin{theorem}\label{theo:1}
  For all graphs $G$ and $H$, the following are equivalent:
  \begin{enumerate}[i]
    \item\label{it:homvec trees} $\vechom_{\cT}(G)=\vechom_{\cT}(H)$.
    \item\label{it:colref} Color refinement does not distinguish $G$ and $H$.
    \item\label{it:fraciso} $G$ and $H$ are fractionally isomorphic, that is, the system $\Fiso(G,H)$ of linear equations has a nonnegative real solution.
  \end{enumerate}
\end{theorem}

As mentioned before, the equivalence between \ref{it:colref} and
\ref{it:fraciso} is due to Tinhofer~\cite{tin86,tin91}. 
An unexpected consequence of our theorem is that we can decide in
time $O((n+m)\log n)$ whether $\vechom_{\cT}(G)=\vechom_{\cT}(H)$ holds for two
given graphs $G$ and $H$ with $n$ vertices and $m$~edges. (If two graphs have a
different number of vertices or edges, then their homomorphism counts
already differ on the 1-vertex or 2-vertex trees.) This is remarkable,
because every known algorithm for computing the entry $\hom(T,G)$ of the vector
$\vechom_{\cT}(G)$ requires quadratic time when~$T$ and~$G$ are given as input.

It is a consequence of the proof of Theorem~\ref{theo:1} that, in order
to characterize an $n$-vertex graph~$G$ up to fractional isomorphisms, it
suffices to restrict the homomorphism vector~$\vechom_{\cT}(G)$ to
trees of height at most $n-1$. What happens if we restrict the structure
of trees even further?
In particular, let us restrict the homomorphism vector
to its path entries, that is, consider~$\vechom_{\cP}(G)$ for the class~$\cP$ of all
paths. Figure~\ref{fig:path} shows an example of two graphs~$G$ and~$H$ with
$\vechom_{\cP}(G)=\vechom_{\cP}(H)$ and $\vechom_{\cT}(G)\neq\vechom_{\cT}(H)$. 

\begin{figure}
  \centering
  \begin{tikzpicture}
    \path[use as bounding box] (120:1.6) -- (-120:1.6) -- (0:1.6);
    \node[bullet] (n)   at (0   :0  ) {};
    \node[bullet] (n1)  at (0   :0.7) {};
    \node[bullet] (n1b) at (0   :1.4) {};
    \node[bullet] (n2)  at (120 :0.7) {};
    \node[bullet] (n2b) at (120 :1.4) {};
    \node[bullet] (n3)  at (-120:0.7) {};
    \node[bullet] (n3b) at (-120:1.4) {};
    \draw[thick] (n) -- (n1) -- (n1b);
    \draw[thick] (n) -- (n2) -- (n2b);
    \draw[thick] (n) -- (n3) -- (n3b);
  \end{tikzpicture}
  \hspace{2cm}
  \begin{tikzpicture}
    \path[use as bounding box] (120:1.6) -- (-120:1.6) -- (0:1.6);
    \node[bullet] (n1)  at (0   :0.7) {};
    \node[bullet] (n2)  at (60  :0.7) {};
    \node[bullet] (n3)  at (120 :0.7) {};
    \node[bullet] (n4)  at (180 :0.7) {};
    \node[bullet] (n5)  at (240 :0.7) {};
    \node[bullet] (n6)  at (300 :0.7) {};
    \node[bullet] (n)   at (0:1.4) {};
    \draw[thick] (n1) -- (n2) -- (n3) -- (n4) -- (n5) -- (n6) -- (n1);
  \end{tikzpicture}
  \caption{Two fractionally non-isomorphic graphs with the same path
    homomorphism counts.}
  \label{fig:path}
\end{figure}
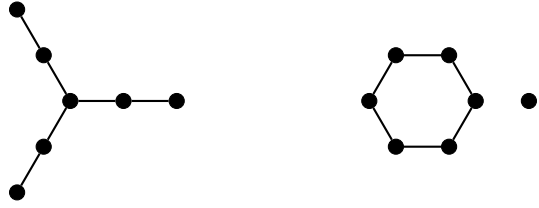

Despite their weaker distinguishing capabilities, the vectors $\vechom_{\cP}(G)$ are quite
interesting. They are related to graph kernels based on counting walks,
and they have a clean algebraic description: it is easy to see that
$\hom(P_k,G)$, the number of homomorphisms from the path $P_k$ of
length $k$ to $G$, is equal to the number of length-$k$ walks in~$G$, which in turn is equal to~$\boldsymbol 1^TA^k\boldsymbol 1$, where~$A$ is
the adjacency matrix of $G$ and $\boldsymbol 1$ is the all-$1$ vector of
appropriate length. 

\begin{theorem}\label{theo:2}
  For all graphs $G$ and $H$, the following are equivalent:
  \begin{enumerate}[i]
    \item\label{it:hom paths}
      $\vechom_{\cP}(G)=\vechom_{\cP}(H)$.
    \item\label{it:Fiso}
      The system $\Fiso(G,H)$ of linear equations has a real solution.
  \end{enumerate}
\end{theorem}

While the proof of Theorem~\ref{theo:1} is mainly graph-theoretic---we
establish the equivalence between the assertions \ref{it:hom paths} and \ref{it:Fiso} by expressing
the ``colors'' of color refinement in terms of specific tree
homomorphisms---the proof of Theorem~\ref{theo:2} is purely
algebraic. We use spectral techniques, but with a twist, because
neither does the spectrum of a graph $G$ determine the
vector $\vechom_{\cP}(G)$ nor does the vector determine the spectrum.
This is in contrast with~$\vechom_{\cC}(G)$ for the class~$\cC$ of all cycles, which, as we already mentioned, distinguishes two graphs if and only if they have the same spectrum.

Let us now turn to homomorphism vectors $\vechom_{\cT_k}(G)$ for the
class $\cT_k$ of all graphs of tree width at most~$k$. We will relate these
to $k$-WL, the $k$-dimensional generalization of color refinement. We
also obtain a
corresponding system of linear equations. Let $G$ and $H$ be graphs
with vertex sets $V$ and $W$, respectively. Instead of variables $X_{vw}$
for vertex pairs ${(v,w)\in V \times W}$, as in the system $\Fiso(G,H)$, the new system
has variables $X_{\pi}$ for $\pi\subseteq V\times W$ of size $|\pi|\le
k$. We call $\pi=\{(v_1,w_1),\ldots,(v_\ell,w_\ell)\}\subseteq V\times W$ a
\emph{partial bijection} if $v_i=v_j\iff w_i=w_j$ holds for all $i,j$, and 
we call it a \emph{partial isomorphism} if in addition
${v_iv_j\in E(G)\iff w_iw_j\in E(H)}$ holds for all~$i,j$.
Now consider the following system~$\Liso[k](G,H)$ of linear equations:

\vspace{-1.5em}
\begin{center}
\begin{minipage}{13cm}
  \begin{empheq}[
      left={\Liso[k](G,H):\quad\empheqlbrace},
      box=\colbox
    ]{align}
    \sum_{v\in V} X_{\pi\cup\{(v,w)\}}&= X_\pi &&
    \parbox[t]{4cm}{for all $\pi\subseteq V\times W$ of size $|\pi|\le k-1$ and all $w\in W$}
    \tag{L1}\label{eq:ck1}\\
    \sum_{w\in W} X_{\pi\cup\{(v,w)\}}&=X_\pi &&
    \parbox[t]{4cm}{for all $\pi\subseteq V\times W$ of size $|\pi|\le k-1$ and all $v\in V$}
    \tag{L2}\label{eq:ck2}\\
    X_\pi&=0 &&
    \parbox[t]{5cm}{for all $\pi\subseteq V\times W$ of size $|\pi|\le k$ such that $\pi$ is not a partial isomorphism from $G$ to $H$}
    \tag{L3}\label{eq:pk}\\
    X_\emptyset&=1
    \tag{L4}\label{eq:ck3}
  \end{empheq}
\end{minipage}
\end{center}

This system is closely related to the Sherali-Adams relaxations of
$\Fiso(G,H)$: Every solution for the level-$k$ Sherali-Adams
relaxation of $\Fiso(G,H)$ yields a solution to $\Liso[k](G,H)$, and
every solution to $\Liso[k](G,H)$ yields a solution to the level $k-1$
Sherali-Adams relaxation of $\Fiso(G,H)$ \cite{atsman13,groott15}.
Our result is this:

\begin{theorem}\label{theo:3}
    For all $k\ge 1$ and for all graphs $G$ and $H$, the following are equivalent:
    \begin{enumerate}[i]
      \item\label{it:hom twk}
        $\vechom_{\cT_k}(G)=\vechom_{\cT_k}(H)$.
      \item\label{it:kWL}
        $k$-WL does not distinguish $G$ and $H$.
      \item\label{it:Liso nonneg}
        $\Liso[k+1](G,H)$ has a nonnegative real solution.
  \end{enumerate}
\end{theorem}

The equivalence between \ref{it:kWL} and \ref{it:Liso nonneg} is implicit in previous work~\cite{immlan90,atsman13,groott15}. 
The system~$\Liso[k](G,H)$ has another nice interpretation related to
the proof complexity of graph isomorphism: it is shown in
\cite{bergro15} that $\Liso[k](G,H)$ has a real solution if and only if a natural system of polynomial
equations encoding the isomorphisms between $G$ and $H$ has a
degree-$k$ solution in the Hilbert Nullstellensatz proof system \cite{beaimpkra+94,bus98}. In view of Theorem~\ref{theo:2}, it is tempting to
conjecture that the solvability of $\Liso[k+1](G,H)$ characterizes the
expressiveness of the homomorphism vectors $\vechom_{\cP_k}(G)$ for the
class $\cP_k$ of all graphs of path width $k$. Unfortunately, we 
only prove one direction of this conjecture.

\begin{theorem}\label{theo:4}
  Let $k$ be an integer with $k\ge 2$ and let $G,H$ be graphs.
  If $\Liso[k+1](G,H)$ has a real solution, then $\vechom_{\cP_k}(G)=\vechom_{\cP_k}(H)$.
\end{theorem}

Combining this theorem with a recent result from \cite{gropak17}
separating the nonnegative from arbitrary real solutions of our
systems of equations, we obtain the following corollary.

\begin{corollary}
  For every $k$, there are graphs $G$ and $H$ with
  $\vechom_{\cP_k}(G)=\vechom_{\cP_k}(H)$ and
  $\vechom_{\cT_2}(G)\neq\vechom_{\cT_2}(H)$.
\end{corollary}

\section{Preliminaries}

\subparagraph*{Basics.}
Graphs in this paper are simple, undirected, and finite
(even though our results transfer to directed graphs and even to
weighted graphs).
For a graph~$G$, we write~$V(G)$ for its vertex set and~$E(G)$ for its edge set.
For $v\in V(G)$, the set of neighbors of~$v$ are denoted with~$N_G(v)$.
For $S\subseteq V(G)$, we denote with~$G[S]$ the subgraph of~$G$ induced by the vertices of~$S$.
A \emph{rooted graph} is a graph~$G$ together with a designated root vertex~$r(G)\in V(G)$.
We write multisets using the notation $\mset{1,1,6,2}$.

\subparagraph*{Matrices.}
An $LU$-decomposition of a matrix~$A$ consists of a lower triangular matrix~$L$ and an upper triangular matrix~$U$ such that $A=LU$ holds.
Every finite matrix~$A$ over~$\R$ has an $LU$-decomposition.
We also use \emph{infinite matrices} over~$\R$, which are functions~$A:I\times J\to\R$ where~$I$ and~$J$ are locally finite posets and countable.
The matrix product~$AB$ is defined in the natural way via $(AB)_{ij} =
\sum_{k} A_{ik} B_{kj}$ if all of these inner products are finite
sums, and otherwise we leave it undefined. 
An $n\times n$ real symmetric matrix has real eigenvalues and a corresponding set of orthogonal eigenspaces.
The spectral decomposition of a real symmetric matrix $M$ is 
of the form $M = \lambda_1 P_1 + \dots + \lambda_l P_l$ where 
$\lambda_1,\dots,\lambda_l$ are the eigenvalues of~$M$ with 
corresponding eigenspaces $W_1,\dots,W_l$. Moreover, 
each~$P_j$ is the projection matrix corresponding to the projection
onto the eigenspace $W_j$. Usually, $P_j$ is expressed 
as $P_j = UU^T$ for a matrix $U$ whose columns form an orthonormal basis of~$W_j$.

\subparagraph*{Homomorphism numbers.}
Recall that a mapping $h:V(F)\to V(G)$ is a homomorphism if $h(e)\in
E(G)$ holds for all $e\in E(F)$ and that $\hom(F,G)$ is the number of homomorphisms from~$F$ to~$G$.
Let $\surj(F,G)$ be the number of homomorphisms from~$F$ to~$G$ that are surjective on both the vertices and edges of~$G$.
Let $\inj(F,G)$ be the number of injective homomorphisms from~$F$ to~$G$.
Let $\sub(F,G)=\inj(F,G)/\aut(F)$, where $\aut(F)$ is the number of
automorphisms of $F$. Observe that $\sub(F,G)$ is the number of
subgraphs of~$G$ that are isomorphic to $F$.
Where convenient, we view the objects $\hom$, $\surj$, and $\inj$ as
infinite matrices; the matrix indices are all unlabeled graphs, sorted
by their size. However, we only use one representative of each
isomorphism class, called the \emph{isomorphism type} of the graphs in
the class, as an index in the matrix.
Then $\surj$ is lower triangular and $\inj$ is upper triangular,
so $\hom = \surj \cdot \sub$ is an LU-decomposition of $\hom$.
Finally, $\ind(F,G)$ is the number of times~$F$ occurs as an induced
subgraph in~$G$. 
Similarly to the homomorphism vectors $\vechom_{\cF}(G)$ we define vectors $\vecinj_\cF(G)$ and $\vecind_\cF(G)$.
Finally, let $G,H$ be rooted graphs. A homomorphism from~$G$ to~$H$ is a graph homomorphism that maps the root of~$G$ to the root of~$H$.
Moreover, two rooted graphs are isomorphic if there is an isomorphism mapping the root to the root.

\section{Homomorphisms from trees}
\label{sec:tree}
\subsection{Color refinement and tree unfolding}
Color refinement iteratively colors the
vertices of a graph in a sequence of \emph{refinement rounds}. Initially, all
vertices get the same color. In each refinement round, any two
vertices~$v$ and~$w$ that still have the same color get different colors if
there is some color $c$ such that~$v$ and~$w$ have a different number
of neighbors of color $c$; otherwise they keep the same color. 
We stop the refinement process if the vertex partition that is induced by the colors does not change anymore, that is, all pairs of
vertices that have the same color before the refinement round still
have the same color after the round. 
More formally, we define the sequence $C^G_0,C^G_1,C^G_2,\ldots$ of colorings as follows.
We let $C^G_0(v) = 1$ for all $v\in V(G)$, and for $i\ge 0$ we let
$C^G_{i+1}(v) = \msetc{C^G_i(u)}{u\in N_G(v)}$.
We say that color refinement \emph{distinguishes} two graphs~$G$ and~$H$ if there is an $i\ge 0$ with
\begin{equation}\label{eq:colref distinguishes}
  \msetc{C_{i}^G(v)}{v\in V(G)}\neq\msetc{C_i^H(v)}{v\in V(H)}\,.
\end{equation}

We argue now that the color refinement algorithm implicitly constructs a tree at~$v$ obtained by simultaneously taking all possible walks starting at~$v$ (and not remembering nodes visited in the past).
For a rooted tree~$T$ with root~$r$, a graph~$G$, and a vertex~$v\in V(G)$, we say
that $T$ is \emph{a tree at~$v$} if there is a homomorphism $f$ from
$T$ to $G$ such that 
    $f(r)=v$ and,
    for all non-leaves $t\in V(T)$, the function~$f$ induces a
    bijection between the set of children of $t$ in in $T$ and the set
    of neighbors of~$f(t)$ in $G$. 
In other words, $f$ is a homomorphism from~$T$ to~$G$ that is \emph{locally bijective}.
If~$T$ is an infinite tree at~$v$ and does not have any leaves,
then~$T$ is uniquely determined up to isomorphisms, and we call this
\emph{the infinite tree at~$v$} (or the \emph{tree unfolding} of~$G$ at~$v$), denoted with~$T(G,v)$.
For an infinite rooted tree~$T$, let~$T_{\le d}$ be the finite rooted subtree of~$T$ where all leaves are at depth exactly~$d$.
For all finite trees~$T$ of depth~$d$, define $\crr(T,G)\in\set{0,\dots,\abs{V(G)}}$ to be the number of vertices~$v\in V(G)$ for which~$T$ is isomorphic to $T(G,v)_{\le d}$.
Note that this number is zero if not all leaves of~$T$ are at the same depth~$d$ or if some node of~$T$ has more than~$n-1$ children.
The $\veccrr$-vector of~$G$ is the
vector~$\veccrr(G)=(\crr(T,G))_{T\in\cT_r}$, where~$\cT_r$ denotes the
family of all rooted trees.
The following connection between the color refinement algorithm and the $\veccrr$-vector is known.

\begin{lemma}[Angluin~\cite{DBLP:conf/stoc/Angluin80}, also see Krebs and Verbitsky~{\cite[Lemma~2.5]{DBLP:conf/lics/KrebsV15}}]\label{lem: colref cr}
  For all graphs $G$ and $H$, color refinement distinguishes $G$ and $H$
  if and only if $\veccrr(G)\neq\veccrr(H)$ holds.
\end{lemma}

\subsection{Proof of Theorem~\ref{theo:1}}
Throughout this section, we work with rooted trees. For a
rooted tree $T$ and an (unrooted) graph $G$, we simply let $\hom(T,G)$
be the number of homomorphisms of the plain tree underlying $T$ to
$G$, ignoring the root.

Let~$T$ and $T'$ be rooted trees.  A homomorphism~$h$ from $T$ to $T'$
is \emph{depth-preserving} if, for all vertices~$v\in V(T)$, the depth
of~$v$ in~$T$ is equal to the depth of~$h(v)$ in~$T'$.  Moreover, a
homomorphism~$h$ from~$T$ to~$T'$ is \emph{depth-surjective} if the
image of~$T$ under~$h$ contains vertices at every depth present
in~$T'$.  We define~$\dhom(T,T')$ as the number of homomorphisms from~$T$ to~$T'$ that are both depth-preserving and depth-surjective.
Note that $\dhom(T,T')=0$ holds if and only if~$T$ and $T'$ have different
depths.

\begin{lemma}\label{lem:hom dhom cr}
  Let $T$ be a rooted tree and let $G$ be a graph.
  We have
  \begin{equation}\label{eq:hom dhom cr}
    \hom( T, G )
    =
    \sum_{T'}
    \dhom(T,T')
    \cdot
    \crr(T',G)
    \,,
  \end{equation}
  where the sum is over all unlabeled rooted trees~$T'$.
  In other words, the matrix identity
  $\hom = \dhom \cdot \crr$ holds.
\end{lemma}
\begin{proof}
  Let~$d$ be the depth of~$T$ and let~$r$ be the root of~$T$.
  Every~$T'$ with $\dhom(T,T')\ne 0$ has depth~$d$ too and there are at most~$n$ non-isomorphic rooted trees~$T'$ of depth~$d$ with $\crr(T',G)\ne 0$.
  Thus the sum in \eqref{eq:hom dhom cr} has only finitely many non-zero terms and is well-defined.
  
  For a rooted tree~$T'$ and a vertex $v\in V(G)$, let $H(T',v)$ be the
  set of all homomorphisms~$h$ from~$T$ to~$G$ such that $h(r)=v$ holds and
  the tree unfolding~$T(G,v)_{\le d}$ is isomorphic to~$T'$.
  Let $H(T')=\bigcup_{v\in V(G)}H(T',v)$ and observe
  $|H(T',v)|=\dhom(T,T')$.
  Since $\crr(T',G)$ is the number of $v\in
  V(G)$ with $T(G,v)_{\le d}\cong T'$, we thus have
  $|H(T')|=\dhom(T,T')\cdot \crr(T',G)$.
  Since each homomorphism from~$T$ to~$G$ is contained in exactly one set~$H(T')$,
  we obtain the desired equality \eqref{eq:hom dhom cr}.
\end{proof}

For rooted trees $T$ and $T'$, let $\dsurj(T,T')$ be the number of
depth-preserving and surjective homomorphisms from~$T$ to~$T'$.  In
particular, not only do these homomorphisms have to be
depth-surjective, but they should hit every vertex of~$T'$.  For
rooted trees $T$ and $T'$ of the same depth, let $\dsub(T,T')$ be
the number of subgraphs of~$T'$ that are isomorphic to~$T$ (under an
isomorphism that maps the root to the root); if $T$ and $T'$ have
different depths, we set $\dsub(T,T')=0$.

\begin{lemma}\label{lem: dhom dsurj sub}
  $\dhom = \dsurj \cdot \dsub$ is an $LU$-decomposition of~$\dhom$, and $\dsurj$ and $\dsub$ are invertible.
\end{lemma}

As is the case for finite matrices, the inverse of a lower
(upper) triangular matrix is lower (upper) triangular. As the matrix
$\dsurj$ is lower triangular and the matrix $\dsub$ is upper
triangular, their inverses are as well.
We are ready to prove our first main theorem.

\begin{proof}[Proof of Theorem~\ref{theo:1}]
  We only need to prove the equivalence between assertions~\ref{it:homvec trees} and~\ref{it:colref}.
For every graph~$G$, let
$\vechom_r(G):=\big(\hom(T,G)\big)_{T\in\cT_r}$.
By our convention that for a rooted tree~$T$ and an unrooted graph $G$
we let $\hom(T,G)$ be the number of homomorphisms of the plain tree underlying $T$ to
$G$, for all $G$ and $H$ we have $\vechom_r(G)=\vechom_r(H)\iff\vechom(G)=\vechom(H)$.
By Lemma~\ref{lem: colref cr}, it suffices to prove for all graph
$G,H$ that
\begin{equation}\label{eq:homcr1}
\veccrr(G)=\veccrr(H)\iff\vechom_r(G)=\vechom_r(H)\,.
\end{equation}
We view the vectors $\vechom_r(G)$ and $\veccrr(G)$ as infinite column vectors.
  By Lemma~\ref{lem:hom dhom cr}, we have
  \begin{equation}
    \vechom_r(G)=\dhom\cdot\veccrr(G)
    \text{ and }
    \vechom_r(H)=\dhom \cdot \veccrr(H)\,.
  \end{equation}
  The forward direction of \eqref{eq:homcr1} now follows immediately.
  
  It remains to prove the backward direction.
  Since $\dhom=\dsurj\cdot\dsub$ holds by Lemma~\ref{lem: dhom dsurj sub} for two invertible matrices $\dsurj$ and $\dsub$, we can first left-multiply with $\dsurj^{-1}$ to obtain the equivalent identities
  \begin{equation}
    \dsurj^{-1}\cdot\vechom_r(G)=\dsub\cdot\veccrr(G)
    \text{ and }
    \dsurj^{-1}\cdot\vechom_r(H)=\dsub\cdot\veccrr(H)\,.
  \end{equation}
  Now suppose $\vechom_r(G)=\vechom_r(H)$ holds, and set $\boldsymbol v=\vechom_r(G)$.
  Then $\dsurj^{-1}\cdot\boldsymbol v$ is well-defined, because
  $\dsurj$ and its inverse are lower triangular. Thus we obtain
  $\dsub\cdot\veccrr(G)=\dsub\cdot\veccrr(H)$ and set $\boldsymbol w=\veccrr(G)$. Unfortunately,
  $\dsub^{-1}\cdot \boldsymbol w$ may be undefined, since $\dsub^{-1}$ is upper triangular.
  While we can
  still use a matrix inverse, the argument becomes a bit subtle. The
  crucial observation is that $\crr(T',G)$ is non-zero for at most~$n$
  different trees~$T'$, and all such trees have maximum degree at most
  $n-1$.
  Thus we do not need to look at \emph{all} trees but only those with maximum degree~$n$.
  Let $\widetilde{\cT}$ be the set of all unlabeled rooted trees of maximum degree at most~$n$.
  Let $\veccrr'=\veccrr|_{\widetilde{\cT}}$,
  let $\boldsymbol w'=\boldsymbol w|_{\widetilde{\cT}}$, and let
  $\dsub'=\dsub|_{\widetilde{\cT} \times \widetilde{\cT}}$.
  Then we still have the following for all $T\in\widetilde{\cT}$ and $G$:
  \begin{equation}
    \boldsymbol w'_{T} = \sum_{T'\in\widetilde{\cT}} \dsub'(T,T')\cdot\crr'(T',G)\,.
  \end{equation}
  The new matrix $\dsub'$ is a principal minor of~$\dsub$ and thus remains invertible.
  Moreover, $\dsub'^{-1}\cdot \boldsymbol w'$ is well-defined, since
  \begin{equation}\label{eq: dsub inverse expanded}
    \sum_{T'\in\widetilde{\cT}} \dsub'^{-1}(T,T')\cdot \boldsymbol w'_{T'}
  \end{equation}
  is a finite sum for each~$T$: The number of (unlabeled) trees $T'\in\widetilde{\cT}$ that have the same depth~$d$ as~$T$ is bounded by a function in~$n$ and~$d$.
  Thus~$\dsub'^{-1}\cdot \boldsymbol w'=\veccrr'(G)$.
  By a similar argument, we obtain $\dsub'^{-1}\cdot \boldsymbol w'=\veccrr'(H)$.
  This implies $\veccrr'(G)=\veccrr'(H)$ and thus $\veccrr(G)=\veccrr(H)$.
\end{proof}

\section{Homomorphisms from cycles and paths}\label{sec:path}
While the arguments we saw in the proof of Theorem~\ref{theo:1} are
mainly graph-theoretic, the proof of Theorem~\ref{theo:2} uses
spectral techniques. 
To introduce the techniques, we first prove a
simple, known result already mentioned in the introduction. We call two square
matrices \emph{co-spectral} if they have the same
eigenvalues with the same multiplicities, and we call two graphs
\emph{co-spectral} if their adjacency matrices are co-spectral.

\begin{prop}[e.g. {\cite[Lemma~1]{van2003graphs}}]\label{prop:spectrum}
  Let $\cC$ be the class of all cycles (including the degenerate cycle of length
  $0$, which is just a single vertex).
  For all graphs $G$ and $H$, we have
  $\vechom_{\cC}(G)=\vechom_{\cC}(H)$
  if and only if $G$ and $H$ are co-spectral.
\end{prop}

For the proof, we review a few simple facts from linear algebra. The
trace $\tr(A)$ of a square matrix $A\in\mathbb R^{n\times n}$ is the sum of the
diagonal entries. If the eigenvalues of
$A$ are $\lambda_1,\ldots,\lambda_n$, then
$\tr(A)=\sum_{i=1}^n\lambda_i$. Moreover, for each $\ell\ge0$ the eigenvalues
of the matrix~$A^\ell$ are $\lambda_1^\ell,\ldots,\lambda_n^\ell$, and thus
$\tr(A^\ell)=\sum_{i=1}^n\lambda_i^\ell$.
The following technical lemma encapsulates the fact that the information $\tr(A^\ell)$ for all~$\ell\in\N$ suffices to reconstruct the spectrum of~$A$ with multiplicities.
We use the same lemma to prove Theorem~\ref{theo:2}, but for Proposition~\ref{prop:spectrum} a less general version would suffice.
\begin{lemma}\label{lem: limit lemma}
  Let $X,Y\subseteq\R$ be two finite sets and let $c\in\R^X_{\ne 0}$ and $d\in\R^{Y}_{\ne 0}$ be two vectors.
  If the equation
  \begin{equation}\label{eq:sumidentity}
    \sum_{x\in X} c_x x^\ell
    =
    \sum_{y\in Y} d_{y} y^\ell
  \end{equation}
  holds for all~$\ell\in\N$,
  then $X=Y$ and $c=d$.
\end{lemma}
\begin{proof}
  We prove the claim by induction on~$k:=\abs{X}+\abs{Y}$.
  For $k=0$, the claim is trivially true since both sums in~\eqref{eq:sumidentity} are equal to zero by convention.

  Let $\hat x=\arg\max\setc{\abs{x}}{x\in X\cup Y}$ and let $\hat x\in X$ without loss of generality.
  If~$\hat x=0$, then $X=\set{0}$ and we claim that~$Y=\set{0}$ holds.
  Clearly \eqref{eq:sumidentity} for~$\ell=0$ yields
  $0\ne c_0 = \sum_{y\in Y} d_y$.
  In particular, $Y\ne\emptyset$ holds. Since~$\hat x=0$ is the maximum of~$X\cup Y$ in absolute value, we have~$Y=\set{0}$ and thus also $c=d$.

  Now suppose that $\hat x\ne 0$ holds.
  We consider the sequences~$(a_\ell)_{\ell\in\N}$
  and $(b_\ell)_{\ell\in\N}$ with
  \begin{align}
    a_\ell
    =
    \frac{1}{\hat x^{\ell}}\cdot
    \sum_{x\in X} c_x x^\ell
    \quad\text{and}\quad
    b_\ell
    =
    \frac{1}{\hat x^{\ell}}\cdot
    \sum_{y\in Y} d_y y^\ell
    \,.
  \end{align}
  Note that $a_\ell=b_\ell$ holds for all $\ell\in\N$ by assumption.
  Observe the following simple facts:
  \begin{enumerate}
    \item[1)]
      If $-\hat x\not\in X$, then
      $\lim_{\ell\to\infty} a_\ell=c_{\hat x}$.
    \item[2)]
      If $-\hat x\in X$, then
      $\lim_{\ell\to\infty} a_{2\ell}=c_{\hat x}+c_{-\hat x}$ and
      $\lim_{\ell\to\infty} a_{2\ell+1}=c_{\hat x}-c_{-\hat x}$.
  \end{enumerate}
  As well as the following exhaustive case distinction for~$Y$:
  \begin{enumerate}
    \item[a)]
      If $\hat x,-\hat x\not\in Y$, then
      $\lim_{\ell\to\infty} b_\ell = 0$.
    \item[b)]
      If $\hat x\in Y$ and $-\hat x\not\in Y$, then
      $\lim_{\ell\to\infty} b_{\ell} = d_{\hat x}$.
    \item[c)]
      If $\hat x\not\in Y$ and $-\hat x\in Y$, then
      $\lim_{\ell\to\infty} b_{2\ell} = d_{-\hat x}$ and
      $\lim_{\ell\to\infty} b_{2\ell+1} = -d_{-\hat x}$.
    \item[d)]
      If $\hat x,-\hat x\in Y$, then
      $\lim_{\ell\to\infty} b_{2\ell} = d_{\hat x}+d_{-\hat x}$ and
      $\lim_{\ell\to\infty} b_{2\ell+1} = d_{\hat x}-d_{-\hat x}$.
  \end{enumerate}
  If~$-\hat x\not\in X$ holds, we see from 1) that~$a_\ell$ converges to the non-zero value~$c_{\hat x}$.
  Since the two sequences are equal, the sequence~$b_\ell$ also converges to a non-zero value.
  The only case for~$Y$ where this happens is b), and we get $\hat x\in Y$, $-\hat x\not\in Y$, and $c_{\hat x}=d_{\hat x}$.
  On the other hand, if~$-\hat x\in X$, we see from 2) that $a_\ell$ does not converge, but the even and odd subsequences do.
  The only cases for~$Y$ where this happens for~$b_\ell$ too are c) and d).
  We cannot be in case c), since the two accumulation points of~$b_\ell$ just differ in their sign, while the two accumulation points of~$a_\ell$ do not have the same absolute value.
  Thus we must be in case d) and obtain $x,\hat x\in Y$ as well as
  \[
    c_{\hat x} + c_{-\hat x}
    =
    d_{\hat x} + d_{-\hat x}
    \quad\text{and}\quad
    c_{\hat x} - c_{-\hat x}
    =
    d_{\hat x} - d_{-\hat x}\,.
  \]
  This linear system has full rank and implies $c_{\hat x}=d_{\hat x}$ and $c_{-\hat x}=d_{-\hat x}$.

  Either way, we can remove $\set{\hat x}$ or $\set{\hat x,-\hat x}$ from both~$X$ and~$Y$ and apply the induction hypothesis on the resulting instance $X',Y',c',d'$.
  Then~$(X,c)=(Y,d)$ follows as claimed.
\end{proof}

\begin{proof}[Proof of Proposition~\ref{prop:spectrum}]
  For all $\ell\ge 0$, the number of homomorphisms from the cycle~$C_\ell$ of length~$\ell$ to a graph $G$ with adjacency matrix~$A$ is equal to the number of closed length-$\ell$ walks in~$G$, which in turn is equal to the trace of $A^\ell$.
  Thus for graphs $G,H$ with adjacency matrices $A,B$, we have $\vechom_{\cC}(G)=\vechom_{\cC}(H)$ if and only if $\tr(A^\ell)=\tr(B^\ell)$ holds for all~$\ell\ge 0$.

  If~$A$ and~$B$ have the same spectrum~$\lambda_1,\dots,\lambda_n$, then
  $\tr(A^\ell)=\lambda_1^\ell+\dots+\lambda_n^\ell=\tr(B^\ell)$ holds for all~$\ell\in\N$.
  For the reverse direction, suppose~$\tr(A^\ell)=\tr(B^\ell)$ for all~$\ell\in\N$.
  Let~$X\subseteq\R$ be the set of eigenvalues of~$A$ and for each~$\lambda\in X$, let~$c_\lambda\in\set{1,\dots,n}$ be the multiplicity of the eigenvalue~$\lambda$.
  Let~$Y\subseteq\R$ and $d_\lambda$ for~$\lambda\in Y$ be the corresponding eigenvalues and multiplicities for~$B$.
  Then for all~$\ell\in\N$, we have
  \[
    \sum_{\lambda\in X} c_\lambda \lambda^\ell
    =
    \tr(A^\ell) = \tr(B^\ell)
    =
    \sum_{\lambda\in Y} d_\lambda \lambda^\ell
    \,.
  \]
  By Lemma~\ref{lem: limit lemma}, this implies~$(X,c)=(Y,d)$, that is, the spectra of~$A$ and~$B$ are identical.
\end{proof}

In the following example, we show that the vectors $\vechom_{\cC}$ for
the class $\cC$ of cycles and
$\vechom_{\cT}$ for the class $\cT$ of trees are incomparable in their
expressiveness. 

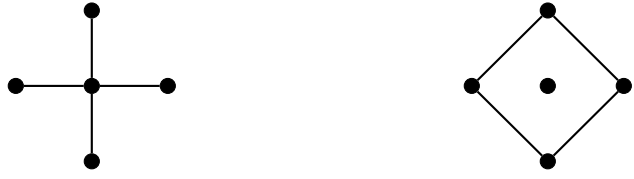
\begin{figure}
  \centering
  \begin{tikzpicture}
    \begin{scope}
      \node[bullet] (a)   at (0   ,0  ) {};
    \node[bullet] (b1)  at (1   ,0) {};
    \node[bullet] (b2)  at (0   ,1) {};
    \node[bullet] (b3)  at (-1   ,0) {};
    \node[bullet] (b4)  at (0   ,-1) {};
    
    \draw[thick] (a) edge (b1) edge (b2) edge (b3) edge (b4);
  \end{scope}
  \begin{scope}[xshift=6cm]
    \node[bullet] (a)   at (0   ,0  ) {};
    \node[bullet] (b1)  at (1   ,0) {};
    \node[bullet] (b2)  at (0   ,1) {};
    \node[bullet] (b3)  at (-1   ,0) {};
    \node[bullet] (b4)  at (0   ,-1) {};
    
    \draw[thick] (b1) edge (b2) (b2) edge (b3) (b3) edge (b4) (b4)
    edge (b1);
    \end{scope}
  \end{tikzpicture}
  \caption{Two co-spectral graphs}
  \label{fig:cospectral}
\end{figure}

\begin{example}
  The graphs $G$ and $H$ shown in Figure~\ref{fig:cospectral} are
  co-spectral and thus $\vechom_{\cC}(G)=
  \vechom_{\cC}(H)$, but it is easy to see that $\vechom_{\cP}(G)\neq
  \vechom_{\cP}(H)$ for the class $\cP$ of all paths.

  Let $G'$ be a cycle of length $6$ and $H'$ the disjoint union of two
  triangles. Then obviously, $\vechom_{\cC}(G')\neq
  \vechom_{\cC}(H')$. However, color refinement does not distinguish
  $G'$ and $H'$ and thus $\vechom_{\cT}(G')=
  \vechom_{\cT}(H')$.
\end{example}

Let us now turn to the proof of Theorem \ref{theo:2}.
\begin{proof}[Proof of Theorem~\ref{theo:2}]
  Let~$A$ and~$B$ be the adjacency matrices of~$G$ and~$H$, respectively.
  Since~$A$ is a symmetric and real matrix, its eigenvalues are real and the corresponding eigenspaces are orthogonal and span~$\R^n$.
  Let~$\boldsymbol 1$ be the~$n$-dimensional all-$1$ vector, and let~$X=\set{\lambda_1,\dots,\lambda_k}$ be the set of all eigenvalues of~$A$ whose corresponding eigenspaces are not orthogonal to~$\boldsymbol 1$.
  We call these eigenvalues the \emph{useful} eigenvalues of~$A$ and without loss of generality assume $\lambda_1>\dots>\lambda_k$.
  The~$n$-dimensional all-$1$ vector~$\boldsymbol 1$ can be expressed as a direct sum of eigenvectors of~$A$ corresponding to useful eigenvalues.
  In particular, there is a unique decomposition
  $\boldsymbol 1=\sum_{i=1}^k u_i$ such that
  each~$u_i$ is a non-zero eigenvector in the eigenspace of~$\lambda_i$.
  Moreover, the vectors~$u_1,\dots,u_k$ are orthogonal.
  For the matrix~$B$, we analogously define its set of useful eigenvalues~$Y=\set{\mu_1,\dots,\mu_{k'}}$ and the direct sum~$\boldsymbol 1=\sum_{i=1}^{k'} v_i$.

  We prove the equivalence of the following three assertions (of which the first and third appear in the statement of Theorem~\ref{theo:2}).
  \begin{enumerate}
  \item $\vechom_{\cP}(G)=\vechom_{\cP}(H)$.
  \item $A$ and $B$ have the same set of useful eigenvalues~$\lambda_1,\dots,\lambda_k$ and $\|u_i\|=\|v_i\|$ holds for all~$i\in\set{1,\dots,k}$.
  Here,~$\|.\|$ denotes the Euclidean norm with~$\|x\|^2=\sum_j x_j^2$.
  \item The system $\Fiso(G,H)$ of linear equations has a
    real solution.
  \end{enumerate}
  Note that in 2, we do not require that the useful eigenvalues occur with the same multiplicities in~$A$ and~$B$.
  We show the implications (1 $\Rightarrow$ 2), (2 $\Rightarrow$ 3), and (3 $\Rightarrow$ 1). 

\medskip
 (1 $\Rightarrow$ 2): Suppose that $\hom(P_\ell,G)=\hom(P_\ell,H)$ holds for all paths~$P_\ell$. Equivalently, this can be stated 
 in terms of the adjacency matrices $A$ and $B$: for all $\ell\in\N$, we have $\trans{\boldsymbol 1} A^{\ell}  \boldsymbol 1 = \trans{\boldsymbol 1} B^{\ell}  \boldsymbol 1$.
 We claim that~$A$ and~$B$ have the same useful eigenvalues, and that the projections of~$\boldsymbol 1$ onto the corresponding eigenspaces have the same lengths.

  Note that~$A^\ell\boldsymbol 1=\sum_{i=1}^k \lambda^\ell_i u_i$ holds.
  Thus we have
  \begin{equation}
    \trans{\boldsymbol 1} A^{\ell} \boldsymbol 1 =
    \paren*{\sum_{i=1}^k \trans{u_i}}
    \paren*{\sum_{i=1}^k \lambda^\ell_i u_i}
    =
    \sum_{i=1}^k \|u_i\|^2 \cdot \lambda^\ell_i
    \,.
  \end{equation}
  The term
  $\trans{\boldsymbol 1} B^{\ell} \boldsymbol 1$ can be expanded analogously, which together yields
  \begin{equation}
    \sum_{i=1}^k \|u_i\|^2 \cdot \lambda^\ell_i
    =
    \sum_{i=1}^{k'} \|v_i\|^2 \cdot \mu^\ell_i
    \quad\text{for all $\ell\in\N$.}
  \end{equation}
  Since all coefficients $c_{\lambda_i}=\|u_i\|^2$ and $d_{\mu_i}=\|v_i\|^2$ are non-zero, we are in the situation of Lemma~\ref{lem: limit lemma}.
  We obtain $k=k'$ and, for all $i\in\set{1,\dots,k}$, we obtain $\lambda_i=\mu_i$ and $\|u_i\|=\|v_i\|$.
  This is exactly the claim that we want to show.

\medskip 
(2 $\Rightarrow$ 3): We claim that the $(n\times n$)-matrix~$X$ defined via
\begin{equation}
   X = \displaystyle\sum_{i=1}^k\frac{1}{\|u_i\|^2} \cdot u_i{v^T_i}
 \end{equation}
satisfies the $\Fiso$ equations $AX=XB$ and $X \boldsymbol 1 = \boldsymbol 1=X^T\boldsymbol 1$.
Indeed, we have
\begin{align}
  AX
  =
  \sum_{i=1}^k\frac{1}{\|u_i\|^2} \cdot 
  A u_i{v^T_i}
  =
  \sum_{i=1}^k\frac{\lambda_i}{\|u_i\|^2} \cdot 
  u_i{v^T_i}
  =
  \sum_{i=1}^k\frac{1}{\|u_i\|^2} \cdot 
  u_i{v^T_i}B^T
  = XB^T
  = XB\,,
\end{align}
This follows, since~$Au_i=\lambda_iu_i$, $Bv_i=\lambda_iv_i$, and $B$ is symmetric.
Moreover, we have
\begin{align}
  X\boldsymbol 1
  =
  \sum_{i=1}^k\frac{1}{\|u_i\|^2} \cdot 
  A u_i{v^T_i}\boldsymbol 1
  =
  \sum_{i=1}^k\frac{1}{\|u_i\|^2} \cdot 
  u_i{v^T_i}\sum_{j=1}^k v_j
  =
  \sum_{i=1}^k\frac{1}{\|u_i\|^2} \cdot 
  u_i \cdot {v^T_i}v_i
  = \boldsymbol 1\,.
\end{align}
This holds by definition of~$u_i$ and~$v_i$ and from $v_i^Tv_i=\|v_i\|^2=\|u_i\|^2$.
The claim $X^T\boldsymbol 1 = \boldsymbol 1$ follows analogously.

\medskip
 (3 $\Rightarrow$ 1): 
Suppose there is a matrix~$X$ with $X^T{\boldsymbol 1} = X{\boldsymbol 1}= {\boldsymbol 1}$ and $AX=XB$.
We obtain $A^{\ell} X = X B^{\ell}$ by induction for all~$\ell\in\N_{>0}$.
For $\ell=0$, this also holds since $A^0=I_n$ by convention.
As a result, we have
$\trans{{\boldsymbol 1}} A^\ell {\boldsymbol 1} =\trans{{\boldsymbol 1}} A^\ell X {\boldsymbol 1} = \trans{{\boldsymbol 1}} X B^\ell {\boldsymbol 1} = \trans{{\boldsymbol 1}} B^\ell {\boldsymbol 1}$ for all~$\ell\in\N$.
Since these scalars count the length-$\ell$ walks in $G$ and $H$, 
respectively, we obtain $\hom(P_\ell,G)=\hom(P_\ell,H)$ for
all paths~$P_\ell$ as claimed.
\end{proof}

\section{Homomorphisms from bounded tree width and path width}

We briefly outline the main ideas of the proofs of Theorems~\ref{theo:3} and \ref{theo:4}; the technical details are deferred to the appendix.
In Theorem~\ref{theo:3}, the equivalence between \ref{it:kWL} and~\ref{it:Liso nonneg} is essentially
known, so we focus on the equivalence between~\ref{it:hom twk} and~\ref{it:kWL}.
The proof is similar to the proof of Theorem~\ref{theo:1} in Section \ref{sec:tree}.

Let us fix $k\ge 2$.
The idea of the $k$-WL algorithm is to iteratively color $k$-tuples of
vertices. Initially, each $k$-tuple $(v_1,\ldots,v_k)$ is colored by
its \emph{atomic type}, that is, the isomorphism type of the labeled
graph $G[\{v_1,\ldots,v_k\}]$. Then in the refinement step, to define
the new color of a $k$-tuple $\bar v$ we look at the current color of
all $k$-tuples that can be reached from $k$ by adding one vertex and
then removing one vertex. 

Similar to the tree unfolding of a graph $G$ at a vertex $v$, we define the \emph{Weisfeiler-Leman tree unfolding} at a $k$-tuple $\bar v$ of vertices.
These objects have some resemblance to the pebbling comonad, which was defined by Abramsky, Dawar, and Wang~\cite{DBLP:conf/lics/AbramskyDW17} in the language of category theory.
The WL-tree unfolding describes the color of $\bar v$ computed by $k$-WL;
formally it may be a viewed as a pair $(T,F)$ consisting of a graph $F$
together with a ``rooted'' tree decomposition (potentially infinite, but again
we cut it off at some finite depth).
Similar to
the numbers $\crr(T,G)$ and the vector $\veccrr(G)$, we now have numbers
$\wl((T,F),G)$ and a vector $\wl(G)$ such that
$\wl(G)=\wl(H)$ holds if and only if $k$-WL does not distinguish~$G$
and $H$. Then we define a linear transformation $\Phi$ with
$\vechom_{\cT_k}(G)=\Phi \wl(G)$.
The existence of this linear transformation directly yields the
implication \ref{it:kWL}$\implies$\ref{it:hom twk} of Theorem~\ref{theo:3}. To prove the
converse, we show that the transformation $\Phi$ is
invertible by giving a suitable $LU$-decomposition of full rank.
This completes our sketch of the proof of Theorem~\ref{theo:3}.

\medskip
The proof of Theorem~\ref{theo:4} requires a different argument,
because now we have to use a solution~$(X_\pi)$ of the system
$\Liso[k+1](G,H)$ to prove that the path width $k$ homomorphism vectors
$\vechom_{\cP_k}(G)$ and $\vechom_{\cP_k}(H)$ are equal. The key idea
is to express entries of a suitable variant of $\vechom_{\cP_k}(G)$ as a linear
combinations of entries of the corresponding vector for $H$ using the
values $X_\pi$ as coefficients.

\section{Conclusions}

We have studied the homomorphism vectors $\vechom_{\cF}(G)$ for
various graph classes $\cF$, focusing on classes $\cF$ where it is
tractable to compute the entries $\hom(F,G)$ of the vector. Our main
interest was in the ``expressiveness'' of these vectors, that is, in
the question what $\vechom_{\cF}(G)$ tells us about the graph
$G$. For the classes $\cC$ of cycles, $\cT$ of trees, $\cT_k$ of graphs of tree width
at most $k$, and $\cP$ of paths, we have obtained surprisingly clean
answers to this question, relating the homomorphism vectors to various
other well studied formalisms that on the surface have nothing to do
with homomorphism counts.

Some interesting questions remain open. The most obvious is whether
the converse of Theorem~\ref{theo:4} holds, that is, whether for two
graphs $G$, $H$ with $\vechom_{\cP_k}(G)=\vechom_{\cP_k}(H)$,
the system $\Liso[k+1](G,H)$ has a real solution (and hence the
Nullstellensatz propositional proof system has no degree-$(k+1)$ refutation of
$G$ and $H$ being isomorphic).

Another related open problem in spectral graph theory is to characterize graphs 
which are identified by their spectrum,
up to isomorphism. In our framework, Proposition \ref{prop:spectrum}
ensures that we can equivalently 
ask for the following characterization: for which graphs $G$ does 
the vector $\vechom_{\cC}(G)$ determine the entire 
homomorphism vector $\vechom(G)$?

Despite the computational intractability, it is also interesting to
study the vectors $\vechom_{\cF}(G)$ for classes $\cF$ of unbounded
tree width. Are there natural classes $\cF$ (except of course the class
of all graphs) for which the vectors $\vechom_{\cF}(G)$ characterize
$G$ up to isomorphism? For example, what about classes of bounded
degree or the class of planar graphs? And what is the complexity of deciding whether
$\vechom_{\cF}(G)=\vechom_{\cF}(H)$ holds when $G$ and
$H$ are given as input? Our results imply that this problem is in polynomial time for the
classes $\cT$, $\cT_k$, and~$\cP$.
For the class of all graphs, it is in quasi-polynomial time by Babai's quasi-polynomial isomorphism test~\cite{bab16}.
Yet it seems plausible that there are classes~$\cF$ (even natural classes decidable in polynomial time) for which the problem is co-NP-hard. 

Maybe the most interesting direction for further research is to study
the graph similarity measures induced by homomorphism
vectors. A
simple way of defining an inner product on the homomorphism vectors is by letting
\[
 \Big\langle  \vechom_{\cF}(G) , \vechom_{\cF}(H)\Big\rangle:=
    \sum_{\substack{k\ge 1\\\cF_k\neq\emptyset}}\frac{1}{k^k|\cF_k|}\sum_{F\in\cF_k}\hom(F,G)\hom(F,H),
  \]
where $\cF_k$ denotes the class of all graph $F\in\cF$ with~$k$ vertices. The mapping $(G,H)\mapsto \langle  \vechom_{\cF}(G) ,
\vechom_{\cF}(H)\rangle$ is what is known as a \emph{graph
  kernel} in machine learning. It induces a (pseudo)metric $d_{\cT}$ on the class
of graphs. It is an interesting question how it relates to other graph
similarity measures, for example, the metric induced by the
Weisfeiler-Leman graph kernel. Our Theorem~\ref{theo:1} implies that
the metric $d_{\cT}$ for the class $\cT$ of trees and the metric
induced by the Weisfeiler-Leman graph kernel have the same graphs of distance zero. 
 
\bibliography{L-Meets-WL}

\newpage
\appendix

\section{Proofs Missing in Section~\ref{sec:tree}}
\label{app:tree}

\subparagraph*{Proof of Lemma~\ref{lem: colref cr}.}
  We devise a bijection~$\pi$ between
  possible colors~$C_i(v)$ and rooted trees~$T$ where each leaf is at
  depth exactly~$i$.  For $i=0$, the only allowed color is~$1$ and, up
  to isomorphism, the only tree is~$T_0$, the tree that only contains
  the root vertex, and so we set $\pi(1)=T_0$.  For $i>0$, let $C$ be
  any color that could appear as the $i$-th round~$C^G_i(v)$ for any
  graph~$G$ and any~$v\in V(G)$.  Then $C$ is a multiset
  $\mset{C_1, \dots, C_\ell}$ consisting of~$\ell$ colors possible to
  create in round~$i-1$.  Let $T_1,\dots,T_\ell$ be rooted trees of
  depth~$i-1$ such that $T_j=\pi(C_j)$ holds for all~$j\in\set{1,\dots,\ell}$.  We define $\pi(C)$ as the
  (unlabeled) rooted tree~$T$ with a new root~$r$ whose
  $\ell$~children are the roots of~$T_1,\dots,T_\ell$.  It is easy to
  see that $\pi$ is a bijection.  Now note that $\pi(C_i(v))$ is
  exactly the isomorphism type~$T$ of the tree~$T(G,v)_{\leq i}$.
  Thus the number $\crr(T, G)$ is equal to the number of
  vertices~$u\in V(G)$ that satisfy~$C_i(u)=C_i(v)$, which proves the
  claim.

\subparagraph*{Proof of Lemma~\ref{lem: dhom dsurj sub}.}
  Let $T$ and $T'$ be rooted trees.
  We want to prove that
  $\sum_{T''}\dsurj(T,T'') \cdot \dsub(T'',T')$
  is well-defined and equal to~$\dhom(T,T')$, where the sum is over all unlabeled rooted trees~$T''$.
  To see that the sum has only finitely many non-zero terms, note that $\dsurj(T,T'')=0$ holds if~$T''$ has more edges or vertices than~$T$, and so the infinite matrix~$\dsurj$ is lower triangular.
  Thus $\dsurj(T,T'')$ is non-zero for finitely many~$T''$.
  
  Note next that $\dsurj(T'',T'')=\aut(T'')\ne 0$ holds for all~$T''$,
  so the lower-triangular matrix $\dsurj$ has nonzero diagonal
  entries, which implies that it is invertible. (This can be seen
  inductively also for infinite matrices, by using forward
  substitution.) 
  Similarly, $\dsub(T'',T')=0$ holds if~$T'$ has fewer edges or vertices than~$T''$, and so the infinite matrix~$\dsub$ is upper triangular.

  Moreover, the diagonal entries satisfy $\dsub(T'',T'')=1$, and so
  the matrix $\dsub$ is invertible as well.

  To prove
  \begin{equation}\label{eq: dsurj dsub dhom expanded}
    \sum_{T''}\dsurj(T,T'') \cdot \dsub(T'',T')=\dhom(T,T')\,,
  \end{equation}
  we devise a bijection~$\pi$ between depth-preserving and depth-surjective homomorphisms~$h$ from~$T$ to~$T'$, and pairs $(h',S)$ where
  \begin{enumerate}
    \item[(i)] $S\subseteq V(T')$ contains the root of~$T'$, at least one deepest leaf of~$T'$, and is connected in~$T'$, and
    \item[(ii)] $h'$ is a depth-preserving and (totally) surjective
      homomorphism from $T$ to~$T''$, where~$T''$ is the isomorphism type of~$T'[S]$.
  \end{enumerate}
  We call~$T''$ the type of the pair~$(h',S)$.  Such a bijection~$\pi$
  implies~\eqref{eq: dsurj dsub dhom expanded}, since pairs~$(h',S)$
  of type~$T''$ can be obtained by choosing one of~$\dsurj(T,T'')$
  possible depth-preserving and surjective~$h'$ from~$T$ to~$T''$ and
  one of~$\dsub(T'',T')$ possible sets~$S$ with the property that
  $T'[S]$ is isomorphic to~$T''$.  Since these choices are
  independent, the number of pairs is equal to the left hand side
  of~\eqref{eq: dsurj dsub dhom expanded}.

  We define $\pi$ as follows. For every $S\subseteq V(T')$ as in (i),
  we fix some isomorphism $\varphi_S$ from $T'[S]$ to $T''$. For
  $h\in\dhom(T,T')$, we let $\pi(h)$ be the pair~$(h',S)$ where $S=h(V(T))$ and $h'=\varphi_S\circ h$.
  Clearly, $h'$ is depth-preserving and surjective from~$T$ to~$T''$.
  Since $h$ is depth-surjective, $S$ contains the root and a deepest leaf of~$T'$, and since~$h$ is a homomorphism and~$T$ is connected, its image~$S$ must also be connected in~$T'$.
  So $\pi$ is a mapping with the correct range, meaning that $(h',S)$ satisfies the two items above.
  To prove that~$\pi$ is injective, note~$\pi(h)\ne\pi(f)$ holds when $h$ and~$f$ have a different image.
  Otherwise they have the same image~$S$ and thus also the same type~$T''$.
   Then $h'=\varphi_S\circ h$ and $f'=\varphi_S\circ f$.
  Since~$\varphi$ is bijective, this implies that~$\pi(h)=\pi(f)$ holds if and only if~$h=f$.
  Finally, to see that $\pi$ is surjective, let $(h',S)$ be any pair from the claimed range of~$\pi$.
  Then $h:=\varphi_S^{-1}\circ h$ satisfies~$\pi(h)=(h',S)$.
  So $\pi$ is a bijection and \eqref{eq: dsurj dsub dhom expanded} holds, which implies the claim.

\section{The Weisfeiler--Leman Algorithm}
\label{app:wl}

Recall that a \emph{partial isomorphism} from a graph $G$ to a graph $H$ is
a set $\pi\subseteq
V(G)\times V(H)$ such that all $(v,w),(v',w')\in\pi$ satisfy the equivalences
$v=v'\iff w=w'$ and $vv'\in E(G)\iff ww'\in E(H)$. We may view $\pi$ as a
bijective mapping from a subset $X\subseteq V$ to a subset of
$Y\subseteq W$ that is an isomorphism from the induced subgraph $G[X]$
to the induced subgraph $H[Y]$.

The \emph{atomic type} $\atp(G,\bar v)$ of a $k$-tuple
$\bar v=(v_1,\ldots,v_k)$ of vertices of a graph $G$ is a description
of the labeled subgraph induced by $G$ on this tuple; formally we may
describe it by a $(k\times k)$-matrix $A$ whose entries satisfy
\begin{equation}\label{eq:atomic type def}
  A_{ij}
  =
  \begin{cases}
    2 & \text{if $i=j$,}\\
    1 & \text{if $i\ne j$ and $v_iv_j\in E$,}\\
    0 & \text{otherwise.}
  \end{cases}
\end{equation}
A crucial property of atomic types is that,
for any two tuples
$\bar v=(v_1,\ldots,v_k)\in V(G)^k$ and $\bar w=(w_1,\ldots,w_k)\in V(H)^k$, we
have $\atp(G,\bar v)=\atp(H,\bar w)$ if and only if
$\pi=\{(v_1,w_1),\ldots,(v_k,w_k)\}$ is a partial isomorphism from $G$
to $H$.
 
\medskip
Let $k\ge 1$.
The $k$-dimensional Weisfeiler-Leman algorithm ($k$-WL) computes a
sequence of \emph{colorings} $C_i^k$ of $V^k$ for a given graph
$G=(V,E)$.
A coloring~$C_i^k$ is \emph{stable} if all $\bar v,\bar w\in V^k$ satisfy
$C_i^k(\bar v)=C_i^k(\bar w)\iff C_{i+1}^k(\bar v)=C_{i+1}^k(\bar w)$.
The output $C_\infty^k$ of~$k$-WL is the coloring with $C_\infty^k=C_i^k$ for the smallest~$i$ such that~$C_i^k$ is stable.

The initial coloring~$C_0^k$ assigns to
each tuple its atomic type: $C_0^k(\bar v):=\atp(\bar v)$. 
In the $(i+1)$st refinement round, the coloring $C_{i+1}^k$
at a tuple $\bar v=(v_1,\ldots,v_k)$
is defined
by
$
C_{i+1}^k(\bar v):=\big(C_i^k(\bar v),M_{i}(\bar v)\big),
$
where $M_i(\bar v)$ is the multiset
{\small
\[
\Big\{\!\!\!\Big\{\big(\atp(v_1,\ldots,v_k,w),C_i^k(v_1,\ldots,v_{k-1},w),C_i^k(v_1,\ldots,v_{k-2},w,v_k),\ldots,C_i^k(w,v_2,\ldots,v_k)\big)\:\Big|\;w\in V
\Big\}\!\!\!\Big\}.
\]}%
If $k\ge 2$ holds, then we can omit the entry
$\atp(v_1,\ldots,v_k,w)$ from the tuples in~$M(\bar v)$, because all
the information it contains is also contained in the entries $C_i^k(\ldots)$ of these tuples.
It is easy to see that the coloring $C_i^1$ computed by $1$-WL coincides with the
colorings $C_i$ computed by color refinement, in the sense that $C_i(v)=C_i(w)\iff
C_i^1(v)=C_i^1(w)$ holds for
all vertices $v,w\in V$.
We say that \emph{$k$-WL distinguishes two graphs~$G$ and~$H$} if
\begin{equation}
  \msetc{C^k_\infty(G;\bar v)}{\bar v\in V(G)^k}
  \ne
  \msetc{C^k_\infty(H;\bar w)}{\bar w\in V(H)^k}
  \,.
\end{equation}

To analyze the strength of $k$-WL as an (incomplete) graph isomorphism test, it is often helpful to use its
characterization as an
equivalence test for the logics $\textsf C^{k+1}$, the $k+1$-variable
fragment of first-order logic with counting quantifiers (see
\cite{gro17}, it is safe to treat the logic and the following lemma as
a black-box here).

\begin{lemma}[\cite{immlan90}]\label{lem:ck}
  For all $k\ge 1$ and all graphs $G$ and $H$ the following are equivalent.
  \begin{enumerate}
  \item $k$-WL distinguishes $G$ and $H$.
  \item The logic $\mathsf C^{k+1}$ distinguishes $G$ and $H$.
  \end{enumerate}
\end{lemma}

\begin{remark}
  Some papers consider a
different version of $k$-WL and count the dimensions differently. If
we denote the version in \cite{bergro15} by $k$-WL$^+$, then $k$-WL distinguishes $G$ and $H$
if and only $(k+1)$-WL$^+$ distinguishes $G$ and $H$. The easiest way
to see this is by going through a logical characterization of the
algorithms (see
  \cite[Section 3.5]{gro17} for details).
\end{remark}

Let us now establish the equivalence between assertions (2) and (3) of
Theorem~\ref{theo:3}. As we mentioned in the introduction, this
equivalence follows easily from the results of \cite{groott15}, but
this may be hard to see for a reader not familiar with that paper. Let
$k\ge 2$, and let~$G$ and~$H$ be graphs with vertex sets $V,W$,
respectively, and adjacency matrices $A,B$, respectively.
In~\cite{groott15}, $k$-WL is characterized in terms of the following
system~$\Fiso[k+1/2](G,H)$ of linear equations in the variables~$X_\pi$ for $\pi\subseteq V\times W$ of size $|\pi|\le k$:

\vspace{-2em}
\begin{center}
  \begin{empheq}[
      left={\Fiso[k+1/2](G,H):\quad\empheqlbrace},
      box=\colbox
    ]{align}
    \sum_{v\in V} X_{\pi\cup\{(v,w)\}} &= X_\pi
    \qquad\text{for all $\pi$ and $w\in W$}
    \label{eq:fhalf1}
    \tag{$\tilde{\text{F}}1$}
    \\
    \sum_{w\in W} X_{\pi\cup\{(v,w)\}}&=X_\pi
    \qquad\text{for all $\pi$ and $v\in V$}
    \label{eq:fhalf2}
    \tag{$\tilde{\text{F}}2$}
    \\
    \sum_{v'}A_{vv'}X_{\pi\cup\{(v',w)\}}&=\sum_{w'}X_{\pi\cup\{(v,w')\}}B_{w'w}
    \quad\parbox[t]{2cm}{for all $\pi,v,w$ with ${\abs{\pi} \lneq k}$}
    \label{eq:fk-1}
    \tag{$\tilde{\text{F}}3$}
    \\
    X_\emptyset&=1
    \tag{$\tilde{\text{F}}$4}
    \label{eq:fhalf3}
  \end{empheq}
\end{center}

\begin{theorem}[\mbox{\cite[Theorem~5.9]{groott15}}]\label{theo:go}
  For all $k\ge 1$ and all graphs $G,H$ the following are equivalent.
  \begin{enumerate}
  \item The logic $\mathsf C^{k+1}$ does not distinguish $G$ and $H$.
  \item $\Fiso[k+1/2](G,H)$ has a nonnegative real solution.
  \end{enumerate}
\end{theorem}

In the following two lemmas, we prove the equivalence between the
systems $\Fiso[k+1/2](G,H)$ and $\Liso[k+1](G,H)$ with respect to
nonnegative solutions. Observe that the two systems have the same
variables $X_{\pi}$ for $\pi\subseteq V\times W$ with $|\pi|\le k+1$,
and they share the
equations \eqref{eq:fhalf1}, \eqref{eq:fhalf2}, \eqref{eq:fhalf3}
(corresponding to \eqref{eq:ck1}, \eqref{eq:ck2}, \eqref{eq:ck3}).

\begin{lemma}\label{lem:LtoF}
  Let $k\ge 2$. Every solution to $\Liso[k+1](G,H)$ is a solution to $\Fiso[k+1/2](G,H)$.
\end{lemma}

\begin{proof}
    Let $(X_\pi)$ be a solution to $\Liso[k+1]$. We need to prove that
    it satisfies the equations \eqref{eq:fk-1}, that is,
    \begin{equation}
      \label{eq:LtoF1}
      \sum_{v'}A_{vv'}X_{\pi\cup\{(v',w)\}}
      =\sum_{w'}X_{\pi\cup\{(v,w')\}}B_{w'w}
    \end{equation}
    for all $\pi\subseteq V\times W$ of size $|\pi|\le k-1$ and all
    $v\in V,w\in W$.

    Let $\pi\subseteq V\times W$ such that
       $|\pi|\le k-1$, and let $v\in V,w\in W$.

    Let $w'\in N(w)$ (that is, $ww'\in E(H)$). Then
    $X_{\pi\cup\{(v,w'),(v',w)\}}=0$ unless $v'\in N(v)$,
    because if $v'\not\in N(v)$ then $\pi\cup\{(v,w'),(v',w)\}$ is not
    a partial isomorphism. Thus by
      \eqref{eq:ck1} applied to $\pi'=\pi\cup\{(v,w')\}$ and $w$,
    \begin{equation}
      \label{eq:3}
      X_{\pi\cup\{(v,w')\}}
      =\sum_{v'\in V}X_{\pi\cup\{(v,w'),(v',w)\}}
      =\sum_{v'\in N(v)}X_{\pi\cup\{(v,w'),(v',w)\}}.
    \end{equation}
    Similarly, for $v'\in N(v)$ we have 
    \begin{equation}
      \label{eq:4}
       X_{\pi\cup\{(v',w)\}}
       =\sum_{w'\in N(w)}X_{\pi\cup\{(v,w'),(v',w)\}}.
    \end{equation}
    These two equations imply \eqref{eq:LtoF1}:
    \begin{align*}
      \sum_{v'\in V}A_{vv'}X_{\pi\cup\{(v',w)\}}&=\sum_{v'\in
                                                  N(v)}X_{\pi\cup\{(v',w)\}}\\
      &=\sum_{v'\in N(v)}\sum_{w'\in
        N(w)}X_{\pi\cup\{(v,w'),(v',w)\}}&\text{by \eqref{eq:4}}\\
      &=\sum_{w'\in
        N(w)}\sum_{v'\in N(v)} X_{\pi\cup\{(v,w'),(v',w)\}}\\
      &=\sum_{w'\in
        N(w)}X_{\pi\cup\{(v,w')\}}&\text{by \eqref{eq:3}}\\
      &=\sum_{w'\in W}X_{\pi\cup\{(v,w')\}}B_{w'w}.
    \end{align*}
\end{proof}

\begin{lemma}[\cite{groott15}]\label{lem:FtoL}
  Let $k\ge 2$. Every nonnegative solution to $\Fiso[k+1/2](G,H)$ is a
  solution to $\Liso[k+1]$.
\end{lemma}

\begin{proof}
  Let $(X_\pi)$ be a nonnegative solution to $\Fiso[k-1/2](G,H)$. We
  need to prove that $\alpha$ satisfies the equations \eqref{eq:pk}
  for $k+1$, that is, $X_\pi=0$ for all $\pi\subseteq
                                                 V\times W$ of size
                                                 $|\pi|\le k+1$ such that
             $\pi$ is not a partial isomorphism from $G$ to $H$.

  \begin{claim}\label{claim:FtoL1}
    For all $\pi'\subseteq\pi\subseteq V\times W$ such that
    $|\pi|\le k+1$, if $X_{\pi'}=0$ then $X_{\pi}=0$.

   \proof
   Clearly, it suffices
   to prove this for the case that $|\pi\setminus\pi'|=1$, say,
   $\pi=\pi'\cup\{(v,w)\}$. Equation~\eqref{eq:ck1} implies that
   \[
     X_{\pi}\le\sum_{v'}X_{\pi'\cup\{(v',w)\}}=X_{\pi'}=0.
   \]
   Note that the first inequality only holds the $X_\pi$ are    nonnegative.
 \end{claim}

 \begin{claim}\label{claim:FtoL2}
   For all $\pi\subseteq V\times W$ such that
    $|\pi|\le k$, if $X_\pi\neq 0$ then $\pi$ is a partial
    bijection.

    \proof
    Let
   $\pi\subseteq V\times W$ such that $|\pi|\le k+1$ and $\pi$ is not
   a partial bijection. Note that there is a $\pi'\subseteq\pi$ of size
   $|\pi'|=2$ such that $\pi'$ is not a partial bijection. By Claim~1, it
   suffices to prove that $X_{\pi'}=0$. Say,
   $\pi'=\{(v,w),(v',w')\}$.

   Suppose first that $v=v'$ and $w\neq w'$. Then
   $X_{\pi'\cup\{(v,w)\}}=X_{\pi'\cup\{(v,w')\}}=X_{\pi'}$. Thus by
   equation~\eqref{eq:ck2} and the nonnegativity 
   \[
     2X_{\pi'}\le\sum_{w''}X_{\pi'\cup\{(v,w'')}=X_{\pi'}.
   \]
   It follows that $X_{\pi'}=0$. Similarly, if $v\neq v'$ and
   $w=w'$ then $X_{\pi'}=0$.
 \end{claim}

 Let
 $\pi\subseteq V\times W$ such that $|\pi|\le k$ and $\pi$ is not a
 partial isomorphism. We need to prove that $X_\pi=0$. 

 Since $\pi$ is not a partial isomorphism, there is a $\pi'\subseteq\pi$ of size
   $|\pi'|=2$ such that $\pi'$ is not a partial isomorphism. By
   Claim~\ref{claim:FtoL1} it suffices to prove that $X_{\pi'}=0$. Say,
   $\pi'=\{(v,w),(v',w')\}$. By
 Claim~\ref{claim:FtoL2}, we may assume that $\pi'$ is a partial bijection, that is, $v\neq v'$ and $w\neq w'$. Then
   $vv'\in E, ww'\not\in F$ or $vv'\not\in E, ww'\in F$. Equivalently,
   $A_{vv'}\neq B_{ww'}$. We look at the instance of
   \eqref{eq:fk-1} for $v,w'$:
   \begin{equation}\label{eq:2}
     \sum_{v''}A_{vv''}X_{\pi'\cup\{(v'',w')\}}=\sum_{w''}X_{\pi'\cup\{(v,w'')\}}B_{w''w'}.
   \end{equation}
   By Claim~\ref{claim:FtoL2}, for all $v''\neq v'$ we have
   $X_{\pi'\cup\{(v'',w')\}}=0$. Similarly, for all $w''\neq w$
   we have $X_{\pi'\cup\{(v,w'')\}}=0$. As
   $\pi'\cup\{(v',w')\}=\pi'\cup\{(v,w)\}=\pi'$, equation \eqref{eq:2}
   reduces to
   \[
     A_{vv'}X_{\pi'}=X_{\pi'}B_{ww'}.
   \]
   As $A_{vv'}\neq B_{ww'}$, it follows that $X_{\pi'}=0$.
\end{proof}

Observe that Lemma~\ref{lem:ck}, Theorem~\ref{theo:go},
Lemma~\ref{lem:LtoF}, and Lemma~\ref{lem:FtoL} imply the equivalence
between assertions (2) and (3) of Theorem~\ref{theo:3}.

\section{Homomorphisms from Small Treewidth}
\label{app:treewidth}
\subsection{More about infinite matrices}

Let $(I,\le)$ be a countable and partially ordered set.
An \emph{interval~$[i,j]$} consists of all~$k\in I$ with~$i\le k\le j$, and the half-open interval~$[i,j)$ is defined as~$[i,j]\setminus\set{j}$.
We assume that the poset $(I,\le)$ is \emph{locally finite}, that is, every interval has finite size.
We call an infinite matrix~$A$ from~$\R^{I\times I}$ upper triangular if~$A_{i,j}=0$ holds for all $i,j$ with $i\not\le j$.
\begin{lemma}
  Let $(I,\le)$ be a locally finite poset.
  If $A\in\R^{I\times I}$ is an upper triangular matrix with $1$s on the diagonal, then the left-inverse $A^{-1}$ is well-defined and upper triangular.
\end{lemma}
\begin{proof}
  We use forward substitution to solve the system~$XA=I$ where $I$ is the identity matrix and~$X$ is going to be the inverse of~$A$.
  We let~$X$ be upper triangular, and for each~$i\in\N$ and $j\in\N$ with $1\le i<j$, we define~$X$ inductively via
  $X_{i,i}=1$ and $X_{i,j}=-\sum_{k\in [i,j)} X_{i,k} A_{k,j}$.

  We verify that~$X$ is indeed the left-inverse of~$A$ by proving that~$XA=I$ holds.
  Indeed, let~$i,j\in I$.
  We have
  \begin{align}
    XA[i,j]
    &=
    \sum_{k} X_{i,k} A_{k,j}\,.
  \end{align}
  Since~$X$ and~$A$ are upper triangular, every term that contributes to the sum satisfies~$i\le k \le j$.
  Since the poset is locally finite, the sum is thus finite and the matrix $XA$ is well-defined.
  In particular, if $i\not\le j$, we have $XA[i,j]=0$.
  Moreover, we have $XA[i,i]=1$.
  Now suppose $i<j$.
  Then
  \begin{align}
    XA[i,j] &=
    X_{i,j} A_{j,j}+
    \sum_{k\in[i,j)} X_{i,k} A_{k,j}
    = 0
  \end{align}
  holds by definition of~$X_{i,j}$.
\end{proof}

\subsection{Strong Homomorphisms}
A homomorphism~$h$ from~$F$ to~$G$ is called \emph{strong} if it also maps non-edges of~$H$ to non-edges of~$G$.
Let $\strhom(F, G)$ be the number of \emph{strong homomorphisms~$h$} from~$F$ to~$G$.
Strong homomorphism numbers turn out to be linear combinations of homomorphism numbers.
To see this, we follow the notation in~\cite{DBLP:conf/stoc/CurticapeanDM17} and define a further counting function $\ext(H,G)$ as follows:
\begin{equation}
  \ext(H,G)=
  \begin{cases}
    0         & \text{if $\abs{V(H)}\ne\abs{V(G)}$, and}\\
    \sub(H,G) & \text{otherwise.}
  \end{cases}
\end{equation}
In particular, $\ext$ inherits its upper triangularity and its~$1$s on the diagonal from~$\sub$.
Moreover, every graph~$H$ has only finitely many graphs~$F$ with~$\ext(H,F)\ne 0$, and so every row and every column of~$\ext$ has finite support.
This implies that~$\ext\cdot A$ is well-defined for any matrix~$A$ of proper dimensions (as opposed to $\sub\cdot A$, which may not be defined if~$A$ has a column of infinite support).
We observe the following matrix identity relating~$\hom$ to~$\strhom$.
\begin{lemma}\label{lem: hom ext strhom}
  We have $\hom=\ext\cdot\strhom$.
\end{lemma}
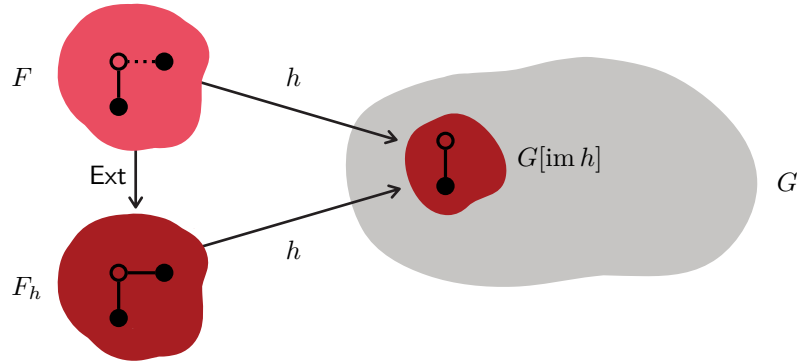
\begin{figure}[tp]
  \centering
  \begin{tikzpicture}[very thick,inner sep=2pt]
    \node {\includegraphics[scale=.5]{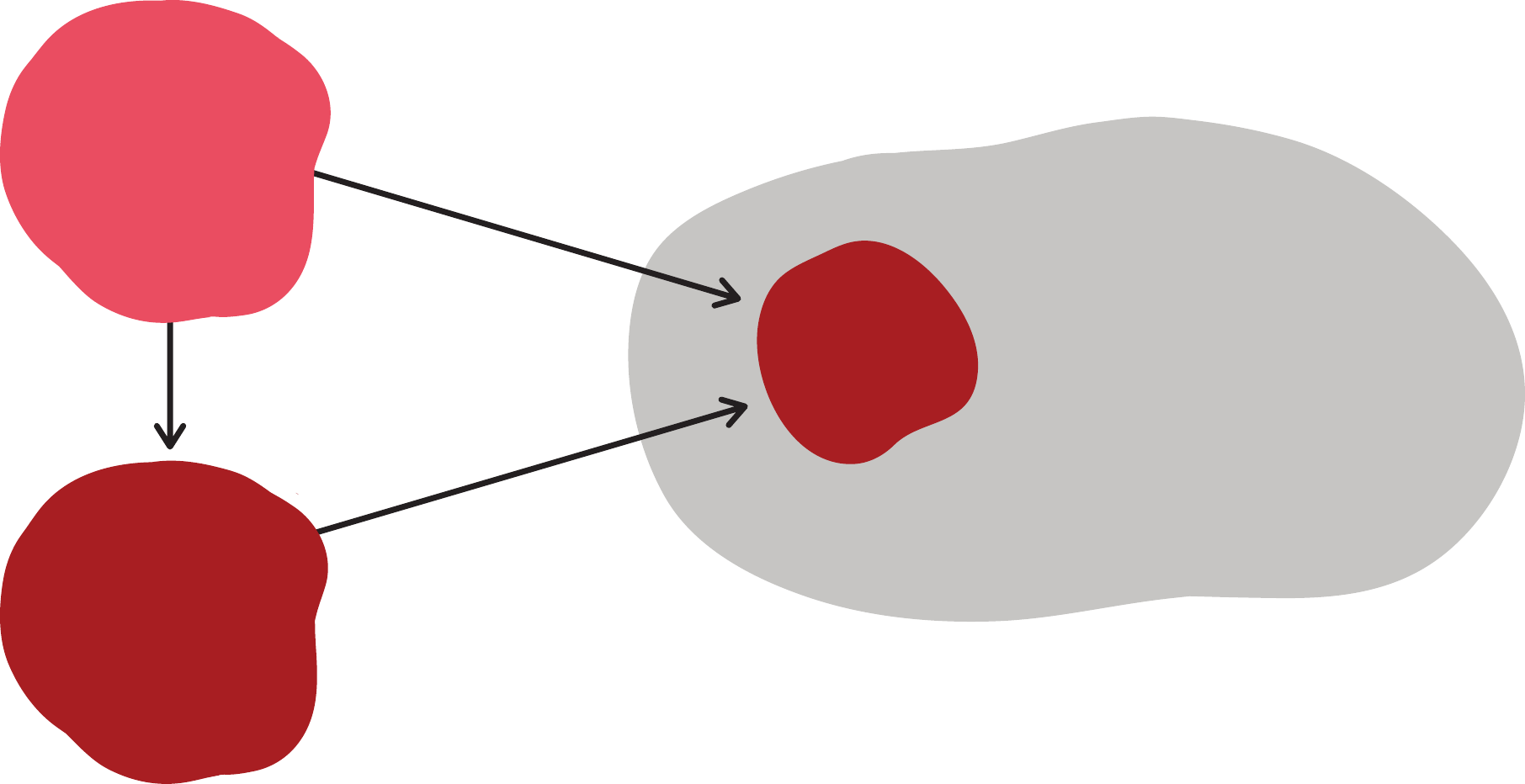}};
    \node[circle,draw] (a2) at (-3.8, 1.6) {};
    \node[circle,draw,fill] (a1) at (-3.2, 1.6) {};
    \node[circle,draw,fill] (a3) at (-3.8, 1.0) {};
    \draw[dotted] (a1) -- (a2);
    \draw (a2) -- (a3);

    \node[circle,draw] (a2) at (-3.8, -1.2) {};
    \node[circle,draw,fill] (a1) at (-3.2, -1.2) {};
    \node[circle,draw,fill] (a3) at (-3.8, -1.8) {};
    \draw (a1) -- (a2) -- (a3);

    \node[circle,draw] (a2) at (0.5, 0.55) {};
    \node[circle,draw,fill] (a3) at (0.5,-0.05) {};
    \draw (a2) -- (a3);

    \node               at ( 5.0, 0.0) {$G$};
    \node[anchor=west]  at (-5.3, 1.4) {$F$};
    \node[anchor=west]  at (-5.3,-1.4) {$F_h$};
    \node[anchor=east]  at (-3.6, 0.1) {$\ext$};
    \node               at (-1.5, 1.4) {$h$};
    \node[anchor=west]  at ( 1.35, 0.3) {$G[\operatorname{im} h]$};
    \node               at (-1.5,-0.9) {$h$};
  \end{tikzpicture}
  \caption{\label{fig:hom-ext-strhom}%
    The proof of Lemma~\ref{lem: hom ext strhom}:
    Every homomorphism~$h$ from~$F$ to~$G$ has a unique extension~$F_h\supseteq F$ such that~$h$ is a strong homomorphism from~$F_h$ to~$G$.
    Depicted are the graph~$F$ (\emph{light shading on the left}), the graph~$G$ (\emph{very light shading on the right}), and the graph~$G[\operatorname{im} h]$ (\emph{dark shading on the right}), which may have fewer vertices (e.g., the two \emph{solid} vertices on the left get mapped to the same \emph{solid} vertex on the right).
    The graph~$F_h$ (\emph{dark shading on the left}) has the same vertex set as~$F$, but gets extend by all possible edges that do not break the homomorphism property of~$h$; doing this ensures that~$h$ becomes a strong homomorphism from~$F_h$ to~$G$.
  }
\end{figure}
\begin{proof}
  Let~$H$ and~$G$ be graphs.
  Let $h$ be a homomorphism from~$H$ to~$G$.
  As depicted in Figure~\ref{fig:hom-ext-strhom}, we define the extension~$H_h$ of~$H$ via the edge relation~$E_{H_h}$ with
  \begin{equation}
    E_{H_h}(u,v) = E_G\paren[\big]{h(u),h(v)} \text{ for all $u,v\in V(H)$}\,.
  \end{equation}
  Then $h$ is a strong homomorphism from~$H_h$ to~$G$ by definition, and~$H_h$ is indeed an extension of~$H$ because~$h$ is a homomorphism from~$H$ to~$G$.
  Moreover, the graph~$H_h$ is the only graph on the vertex set~$V(H)$ such that~$h$ is a strong homomorphism from it.
  Thus we have established a bijection between homomorphisms~$h$ from~$H$ to~$G$ and pairs~$(H',h)$ where~$H'$ is an extension of~$H$ and~$h$ is a strong homomorphism from~$H'$ to~$G$.
  This implies
  \begin{equation}\label{eq: hom ext strhom expanded}
    \hom(H,G)
    =
    \sum_{\substack{H'\supseteq H\\V(H')=V(H)}}
    \strhom(H',G)
    =
    \sum_F
    \ext(H,F)\cdot\strhom(F,G)
    \,,
  \end{equation}
  where the second equality follows by collecting terms for isomorphic graphs~$H'$.
  We arrive at the claimed matrix identity.
\end{proof}

\subsection{Homomorphisms with bag-wise properties}
Let $k$ be a fixed positive integer.
Let $T$ be a width-$k$ tree decomposition of a finite undirected graph~$F$. Recall that $T$ is a rooted tree whose bags are sets $\beta(t)\subseteq V(F)$ for all nodes~$t\in V(T)$.
We further assume that all bags are distinct, all bags at even depths (including the root) have size~$k$, and all bags at odd depths have size~$k+1$.
The pair~$(F,T)$ is called a tree-decomposed graph.
\begin{definition}\label{def:hom variants}
  We define homomorphism numbers for tree-decomposed graphs~$(F,T)$:%
\begin{enumerate}
  \item
    $\hom( (F,T), G) = \hom(F,G)$ is the number of homomorphisms from~$F$ to~$G$.
  \item
    $\biso( (F,T), G)$
    is the number of homomorphisms~$h$ from~$F$ to~$G$ such that, for all $t\in V(T)$, the mapping~$h:\beta(t)\to V(G)$ is an isomorphism from~$F[\beta(t)]$ to~$G[h(\beta(t))]$.
    That is, $\biso((F,T),G)$ is the number of homomorphisms that are bag-wise isomorphisms.
  \item
    $\binj( (F,T), G)$
    is the number of homomorphisms $h$ from~$F$ to~$G$ such that, for all $t\in V(T)$, the mapping~$h:\beta(t)\to V(G)$ is injective.
    That is, it counts bag-wise injective homomorphisms.
  \item
    $\bstrhom( (F,T), G)$
    is the number of homomorphisms~$h$ from~$F$ to~$G$ such that, for all $t\in V(T)$, the mapping~$h:\beta(t)\to V(G)$ is a strong homomorphism from~$F[\beta(t)]$ to~$G[h(\beta(t))]$, that is, it also maps non-edges of bags to non-edges of~$G$.
  \item
    $\dbiso( (F,T), (F',T'))$
    is the number of homomorphisms~$h$ from~$F$ to~$F'$ such that, for all $t\in V(T)$, the set~$h(\beta(t))$ is equal to a bag~$\beta(t')$ of~$T'$, the mapping $h$ is an isomorphism from $F[\beta(t)]$ to $F'[\beta(t')]$, and the corresponding mapping from~$V(T)$ to~$V(T')$ is a depth-preserving and depth-surjective homomorphism from~$T$ to~$T'$.
    Similarly, $\dbsurj( (F,T), (F',T'))$ counts~$h$ only if the latter mapping is depth-preserving and surjective from~$T$ to~$T'$, and
    $\dbsub( (F,T), (F',T'))$ counts~$h$ only if the mapping is an injective homomorphism of~$T$ in~$T'$ with the property that~$T$ and~$T'$ have the same depth.
  \item An isomorphism from $(F,T)$ to $(F',T')$ is an isomorphism~$h$ from~$F$ to~$F'$ such that the corresponding mapping from~$V(T)$ to~$V(T')$ induced on the bags is an isomorphism from~$T$ to~$T'$.
  \item $\bext( (F,T), (F',T'))$ is the number of \emph{bag-wise extension of~$(F,T)$} isomorphic to~$(F',T')$, which are graphs~$(H,T)$ with~$V(H)=V(F)$ such that $H[\beta(t)]$ is an extension of~$F[\beta(t)]$ for every node~$t\in V(T)$.
\end{enumerate}
\end{definition}

We also need a partial order~$\le$ on the set of all tree-decomposed graphs.
We let this be the partial order induced by the lexicographic order on the tuple~$(w,d,n+m)$ computed from~$(F,T)$ by relying on the width~$w$ of~$T$, the maximum degree~$d$ of~$T$, the number~$n$ of vertices of~$F$, and the number~$m$ of edges of~$F$.

We will prove some matrix identities for these matrices, as they are used in the proof of our main result.
We start with a matrix identity for $\bstrhom$ analogous to Lemma~\ref{lem: hom ext strhom}.
For this, we introduce bag-wise extensions.
Note that every row and every column of~$\bext$ has finite support.
\begin{lemma}\label{lem:bhom bext bstrhom}
  We have $\hom = \bext \cdot\bstrhom$.
\end{lemma}
\begin{proof}
  The proof is analogous to the proof of Lemma~\ref{lem: hom ext strhom}.
  Let~$(H,T)$ and~$(G,T)$ be tree-decomposed graphs.
  Let $h$ be a homomorphism from~$H$ to~$G$.
  We define the extension~$H_h$ of~$H$ via the edge relation~$E_{H_h}$ that satisfies the following for all~$u,v\in V(H)$:
  \begin{equation}
    E_{H_h}(u,v) =
    \begin{cases}
      E_G\paren[\big]{h(u),h(v)} & \text{if $u$ and~$v$ co-occur in some bag of~$T$,}\\
      E_H\paren[\big]{h(u),h(v)} & \text{otherwise.}
    \end{cases}
  \end{equation}
  Equivalently, to obtain~$H_h$ from~$H$, we add edges between any two non-adjacent vertices~$u$ and~$v$ that occur together in the some bag~$\beta(t)$ of~$T$ and whose image $h(u)h(v)$ forms an edge in~$G$.
  Non-edges of~$H$ that do not occur in any bag remain non-edges in~$H_h$.

  By construction, $h$ is a bag-wise strong homomorphism from~$H_h$ to~$G$, and~$H_h$ is indeed an extension of~$H$.
  Moreover, $H_h$ is the only graph on the vertex set~$V(H)$ such that~$T$ remains a tree decomposition for~$H_h$ and~$h$ is a bag-wise strong homomorphism from $(H_h,T)$ to~$G$.
  Thus we have established a bijection between homomorphisms~$h$ from~$(H,T)$ to~$G$ and pairs~$(H',h)$ where~$(H',T)$ is a bag-wise extension of~$(H,T)$ and~$h$ is a bag-wise strong homomorphism from~$(H',T)$ to~$G$.
  This implies
  \begin{align}
    \hom((H,T),G)
    &=
    \sum_{\substack{H'\\\text{$(H',T)$ extends $(H,T)$ bag-wise}}}
    \bstrhom((H',T),G)\\
    &=\label{eq: bhom bext bstrhom expanded}
    \sum_{(F,T')}
    \bext((H,T),(F,T'))\cdot\bstrhom((F,T'),G)
    \,,
  \end{align}
  where the second equality follows by collecting terms for isomorphic tree-decomposed graphs~$(H',T)$ and the sum is over all isomorphism types of tree-decomposed graphs~$(F,T')$.
  Since~$(H,T)$ has only finitely many bag-wise extensions, the sums are indeed finite.
  We arrive at the claimed matrix identity.
\end{proof}
If~$F$ is a tree, we can choose a tree decomposition~$T$ of width~$1$ whose bags don't contain any non-edges, and so $\bstrhom(F,T)$ is equal to~$\hom(F,T)$.
This explains why we did not have to deal with strong homomorphisms in the proof of Theorem~\ref{theo:1}, where we established the equivalence between color refinement and homomorphism numbers from trees.

The next matrix identity is an analogue of Lemma~\ref{lem: dhom dsurj sub}.
\begin{lemma}
  $\dbiso=\dbsurj\cdot\dbsub$ is an $LU$-decomposition and $\dbsurj$ and $\dbsub$ are invertible.
\end{lemma}

\subsection{Weisfeiler--Leman tree unfoldings and homomorphisms}
Recall that the \emph{atomic type} $\atp(G,\bar v)$ of
a $k$-tuple $\bar v=(v_1,\ldots,v_k)$ of vertices of a graph~$G$ is a description
of the labeled subgraph induced by $G$ on this tuple; formally
we may describe it by a $(k\times k)$-matrix $A$ with entries
$A_{ij}=2$ if $v_i=v_j$ and $A_{ij}=1$ if $v_iv_j\in E$ and $A_{ij}=0$
otherwise.
A matrix $A\in\set{0,1,2}^{k\times k}$ is called an \emph{atomic type} if it is symmetric and has $2$'s on the diagonal.

We define the Weisfeiler-Leman tree unfolding of a graph, which can be viewed as the strategy tree of an Ehrenfeucht–Fraïssé game with~$k$ pebbles.
\begin{definition}\label{def: WL unfolding}
  Let $k$ be a positive integer, let $G$ be a graph, and let $v_1,\dots,v_k\in V(G)$ be distinct vertices.
  The \emph{WL-tree unfolding at $(v_1,\dots,v_k)$} is the tree-decomposed graph~$(F,T)$ that is constructed together with a bag-wise isomorphic homomorphism~$\pi$ from~$(F,T)$ to~$G$ as follows:
  \begin{enumerate}
    \item
      We start with $F$ having~$k$ vertices~$1,\dots,k$ and~$T$ being the trivial tree
      decomposition with a single bag~$\beta(t):=V(F)$ at the unique node~$t\in
      V(T)$.
      Let $\pi: V(F)\to \set{v_1,\dots,v_k}$ be the function with $\pi(i)=v_i$ for $i\in\set{1,\dots,k}$.
      Let the edges of $F$ be chosen such that $\pi$ is an isomorphism from $F$ to~$G[\set{v_1,\dots,v_k}]$.
    \item\textit{(Introduce nodes)}
      If $t$ is a leaf node of~$T$ with $\abs{\beta(t)}=k$, then for each $w\in V(G)$ with~$w\not\in \pi(\beta(t))$, we do the following:
      \begin{enumerate}
        \item
          Add a fresh child node~$t_w$ to~$t$ in~$T$.
        \item
          Add a fresh vertex~$f$ to~$F$ and extend~$\pi$ with $[f\mapsto w]$
        \item
          Let the bag of~$t_w$ be defined via
          $\beta(t_w)=\beta(t)\cup\set{f}$.
        \item
          Add edges between~$f$ and~$\beta(t)$ to~$F$ in the unique way so that~$\pi$ is an isomorphism from~$F[\beta(t_w)]$ to~$G[\pi(\beta(t_w))]$.
      \end{enumerate}
    \item\textit{(Forget nodes)}
      If $t$ is a leaf node of~$T$ with $\abs{\beta(t)}=k+1$, then for each $j\in\set{1,\dots,k}$, we do the following:
      \begin{enumerate}
        \item
          Add a fresh child node~$t_j$ of~$t$ to~$T$.
        \item
          Let $f$ be the vertex that was introduced at~$t$, that is, we have $\beta(t')\setminus\beta(t)=\set{f}$ for the parent~$t'$ of~$t$.
        \item 
          Let $\set{f_1,\dots,f_k}=\beta(t)\setminus\set{f}$ where the $f_i$ are sorted in a canonical way.
        \item
          We define $\beta(t_j):=\beta(t)\setminus\set{f_j}$.
      \end{enumerate}
  \end{enumerate}
  Clearly, applying rules (2) and (3) above a finite number of times constructs a tree-decomposed graph~$(F,T)$ and a bag-wise isomorphic homomorphism~$\pi$.
  If we exhaustively expand leaves of~$T$ at depth less than~$d$ and then stop the process, all leaves of the final tree~$T$ are at depth~$d$ and we say that $(F,T)$ is \emph{the depth-$d$ WL-tree unfolding of~$G$ at~$\bar v$}.

  For all tree-decomposed graphs~$(F,T)$, let
  $\wl( (F,T), G)$ be the number of tuples $\bar v:= (v_1,\dots,v_k)$ of vertices in~$G$ for which the WL-tree unfolding~$(F',T')$ at $\bar v$ is isomorphic to~$(F,T)$.
\end{definition}

We remark already here that, if $(T,F)$ has leaves at different depths or if it has non-leaves with more than~$n$ children, then $\wl((F,T),G)=0$ holds.
We now state the equivalence between the $k$-dimensional Weisfeiler--Leman algorithm and the homomorphism numbers from treewidth-$k$ graphs.
\begin{theorem}\label{thm:high-dim}
  Let $k$ be a positive integer, and let $G$ and $H$ be finite undirected graphs.
  Then the following are equivalent:
  \begin{enumerate}
    \item\label{item: WL does not distinguish}
      $C_\infty^k(G)=C_\infty^k(H)$
    \item\label{item: WL-unfolding does not distinguish}
      For all finite undirected graphs~$F$ with a tree decomposition~$T$ of width at most~$k$, we have $\wl( (F,T), G) = \wl( (F,T), H)$.
    \item\label{item: bisovector does not distinguish}
      For all finite undirected graphs~$F$ with a tree decomposition~$T$ of width at most~$k$, we have $\biso( (F,T), G) = \biso( (F,T), H)$.
    \item\label{item: homvector does not distinguish}
      For all finite undirected graphs~$F$ with $\tw(F)\le k$, we have $\hom( F, G) = \hom( F, H)$.
  \end{enumerate}
\end{theorem}
\begin{proof}
  ``\ref{item: WL does not distinguish} is equivalent to \ref{item: WL-unfolding does not distinguish}'':
  The proof is almost entirely syntactical, and a generalization of the proof of Lemma~\ref{lem: colref cr}, which establishes the case~$k=1$.
  In particular, (a) the object~$C^k(G,\bar v)$ constructed by the $k$-dimensional Weisfeiler--Leman algorithm implicitly constructs a WL-tree unfolding at~$\bar v$, and (b) from the WL-tree unfolding at~$\bar v$, we can reconstruct the entire object~$C^k(G,\bar v)$.
  These two claims imply the equivalence.
  For claim~(b), we define the object~$\tilde C^k_i(G,\bar v)$ modified from~$C^k_i(G,\bar v)$ in such a way that we only recurse on tuples that contain \emph{distinct} vertices.
  That is,
  $\tilde C^k_{i+1}(G,v_1,\dots,v_k)$
  is defined as
  \begin{equation}
    \msetc{
    \big(
    \atp(G,\bar v,w),
    \tilde C^k_i(v_1,\dots,v_{k-1},w),
    \dots,
    \tilde C^k_i(w,v_2,\dots,v_{k})
  \big)
    }{
    w\in V\setminus\set{v_1,\dots,v_k}
    }
  \end{equation}
  All information about~$\tilde C^k_i$ is contained in~$C^k_i$, since we can simply ignore atomic types that contain $2$s. Conversely, the object~$C^k_i$ can be reconstructed from~$\tilde C^k_i$ by recursively adding elements for~$w\in\set{v_1,\dots,v_k}$.
  Claim~(a) follows directly because~$\tilde C^k_i(G,\bar v)$ is just a different notation for the depth-$2i$ WL-tree unfolding at~$\bar v$.

  ``\ref{item: WL-unfolding does not distinguish} is equivalent to \ref{item: bisovector does not distinguish}'':
  In analogy to Lemma~\ref{lem:hom dhom cr}, we have the following identity:
  \begin{equation}
    \biso((F,T), G)
    =
    \sum_{(F',T')} \dbiso((F,T), (F',T')) \cdot \wl((F',T'),G)\,.
  \end{equation}
  To prove the identity, let $(F,T)$ be a tree-decomposed graph such that~$T$ has depth~$d$.
  (Note that~$T$ may have leaves at different depths.)
  The sum is over all isomorphism types~$(F',T')$ of tree-decomposed graphs.
  Since $\dbiso((F,T),(F',T'))=0$ holds if~$T'$ has depth~$>d$ or nodes with $>n$ children, the sum is finite and thus well-defined.

  Consider a bag-wise isomorphic homomorphism~$h$ from~$(F,T)$ to~$G$.
  Let $1,\dots,k$ be the vertices in the root bag of~$T$, and let~$v_i=h(i)$ for all $i\in\set{1,\dots,k}$.
  Let~$(F',T')$ be the WL-tree unfolding of depth~$d$ at~$\bar v$ in~$G$.
  Let~$\pi:V(F')\to V(G)$ be the bag-wise isomorphic homomorphism constructed during this unfolding.
  Now $h$ can be split into two steps: First, map $(F,T)$ to $(F',T')$ using a homomorphism~$\sigma$, then map into~$G$ using~$\pi$.
  To define~$\sigma$, we map the first bag of~$T$ to the first bag of~$T'$ in order.
  We continue inductively:
  If~$t$ is a node of~$T$ at an even depth, ~$\sigma(t)=t'$ holds, and $r$ is a child of~$T$ where a vertex~$f$ is introduced.
  Then let~$r'$ be the child of~$t'$ in~$T'$ where the vertex~$\pi_{r'}^{-1}(h(f))$ is introduced.
  Forget bags are analogous, and the mapping~$\sigma$ constructed in this way is bag-wise isomorphic and depth-surjective.
  Since the objects~$h$ and $(\sigma,\bar v)$ are in one-to-one correspondence, the claimed identity follows.
  The matrix~$\dbiso$ is invertible, for which reason the claimed equivalence ``$2\Leftrightarrow 3$'' of the Theorem follows.
  
  ``\ref{item: bisovector does not distinguish} is equivalent to \ref{item: homvector does not distinguish}'':
  Let $F$ be a graph and let~$T$ be a width-$k$ tree-decomposition of~$F$.
  We have the following identity:
  \begin{equation}
    \hom((F,T), G)
    =
    \sum_{(F',T')} (\surj\cdot\bext)(F',T')\cdot\biso((F',T'),G)/\aut(F',T')\,.
  \end{equation}
  Here, $\surj$ is the number of all homomorphisms from~$F$ to $F'$ that are vertex- and edge-surjective, such that every bag $\beta(t)$ for $t\in V(T)$ is mapped to a bag $\beta(t')$ for $t'\in V(T')$, and the latter mapping is a surjective homomorphism from~$T$ to~$T'$.
  This matrix~$\surj$ is invertible.
  Moreover, $\aut(F',T')$ is equal to the number of isomorphisms from~$(F',T')$ to~$(F',T')$.
  Writing $\aut$ as a diagonal matrix, equation corresponds to the matrix identity
  $\hom=\surj\cdot\bext\cdot\aut^{-1}\cdot\biso$.
  Since the matrices~$\surj$, $\bext$, $\aut^{-1}$ are invertible when restricting them to the finite submatrices whose indices~$(F',T')$ have depth at most~$d$, we obtain the equivalence claimed by the theorem.
\end{proof}

\section{Homomorphisms from Small Pathwidth}
\label{app:pathwidth}

Let $k$ be a fixed positive integer. Let $P$ be width-$k$ path decomposition
of a finite undirected graph $F$, where $P=(X_1,Y_1,X_2,\dots,X_{\ell})$. 
Here, $\vert X_1 \vert = \dots = \vert X_\ell \vert = k$, 
and $\vert Y_1 \vert = \dots = \vert Y_{\ell-1}\vert = k+1$. 
Also, $\ell$ is the length parameter of the decomposition. 
We define a conditional variant of $\biso((F,P),G')$ as follows. 
Given a graph $G$, let $\bisocond{(F,P),G}{\begin{smallmatrix} u_1 \dots u_k \\  v_1 \dots v_k  \end{smallmatrix}}$ denote 
the number of bag-wise isomorphic homomorphisms from $F$ to $G$ which, in addition, map the vertices $u_1,\dots,u_k \in V(F)$ 
to $v_1,\dots,v_k \in V(G)$ respectively.

Let us now fix graphs $G$ and $H$ with vertex sets $V,W$,
respectively, such that the system
$\Liso[k+1](G,H)$ has a real solution $(X_\pi)$, where $\pi$ ranges
over all subsets of $V\times W$ of size at most $k+1$.

The following lemma shows how to ``transfer'' the conditional
bag-wise isomorphic homomorphism numbers 
across the graphs $G$ and $H$.
 
\begin{lemma}\label{lem:transfer}
Let $F$ be a finite undirected graph with a path decomposition $P$ of width $k$,
where $P=(X_1,Y_1,X_2,\dots,X_{\ell})$. Let $X_1 =\{u_1,\dots,u_k\}$.
Then  for all $(v_1,\dots,v_k) \in V^k$, 
\begin{align*}
  \bisocond{(F,P),G}{\begin{smallmatrix}
                      u_1 \dots u_k \\
                      v_1 \dots v_k
                     \end{smallmatrix}
                     } = \displaystyle\sum_{\substack{(w_1,\dots,w_k) \\ \in W^k}} X_{ \{(v_1, w_1),\dots, (v_k,w_k)\} } \,\bisocond{(F,P),H}{ \begin{smallmatrix}
                      u_1 \dots u_k \\
                      w_1 \dots w_k
                     \end{smallmatrix} }
\end{align*}
\end{lemma}
\begin{proof}
 The proof is by induction on the length parameter $\ell$. 
 The base case $\ell = 1$ corresponds to the situation when $P$ consists of a single bag $X_1$. 
 Denote $\tau = \atp(F,(u_1,\dots,u_k))$. 
 Clearly, $\bisocond{(F,P),G}{\perm{v}{1}{k}} = 1$ if $\atp(G,(v_1,\dots,v_k)) =\tau$, and zero otherwise.
 Likewise, $\bisocond{(F,P),H}{\perm{w}{1}{k}}$ is $1$ if $\atp(H,(w_1,\dots,w_k)) =\tau$, and zero otherwise.
 Therefore, given a $k$-tuple $\bar{v} =(v_1,\dots,v_k) \in V^k$, there are two possibilities: 
 either $\atp(G,(v_1,\dots,v_k)) \neq \tau$. Then, the LHS is zero. The RHS is also zero 
 since $X_{\bar{v} \mapsto \bar{w}} = 0$ for every $(w_1,\dots,w_k)$ of type $\tau$. Otherwise, the 
 second possibility is that $\atp(G,(v_1,\dots,v_k)) = \tau$. Then, the LHS $\bisocond{(F,P),G}{\perm{v}{1}{k}} = 1$. 
 The RHS sum can be taken over all $\bar{w}$ such that $\atp(H,\bar{w}) =\tau$. The RHS simplifies to 
 the sum of all $X_{\bar{v} \mapsto \bar{w}}$ over all $\bar{w}$ of atomic type $\tau$. 
 This sum is equal to $1$, a fact which is immediate from repeated application of equations \ref{eq:ck1} - \ref{eq:ck3} of $\Liso[k+1](G,H)$.
 This finishes the base case. 
 
 We proceed to the inductive case for $\ell>1$.
 In the path decomposition $P$, let $Y_1 = \{u_0,\dots,u_{k}\}$ and $X_2=\{u_0,\dots,u_{k-1}\}$ 
 (in usual terminology, we say that we ``introduce'' the vertex $u_0$ in bag $Y_1$ and ``forget'' the vertex $u_k$ in $X_2$).
 Let $\tilde{P}= (X_2,\dots,Y_{\ell-1},X_\ell)$ be the corresponding path decomposition for the graph $\tilde{F}$, where
 $\tilde{F}$ is the graph $F$ with vertex $u_k$ deleted. The length parameter of the decomposition $\tilde{P}$ is $\ell -1$.
 We can rewrite
\begin{align*}
 \bisocond{(F,P),G}{\perm{v}{1}{k}} &= \displaystyle\sum_{v_0 \in V} \bisocond{(F,P),G}{\perm{v}{0}{k}} \\
 &= \displaystyle\sum_{v_0 \in V}  \bisocond{ (\tilde{F},\tilde{P}),G}{\perm{v}{0}{k-1}} \cdot I_{G,[k-1]} \cdot I_{G,\{0\}} 
\end{align*}
where $I_{G,[k-1]}=1$ if the adjacency of $u_k$ to $u_1,\dots,u_{k-1}$ in $F$ is identical to the adjacency of $v_k$ to $v_1,\dots,v_{k-1}$ in $G$, 
and is zero otherwise. 
Likewise, $I_{G,\{0\}}=1$ if the adjacency $\{u_k,u_0\}$ in $F$ is equal to the adjacency $\{v_k,v_0\}$ in $G$, and is zero otherwise.  
Using the inductive hypothesis for $(\tilde{F},\tilde{P})$, we rewrite
\begin{align*}
  \begin{split} \bisocond{(F,P),G}{\perm{v}{1}{k}} & = \displaystyle\sum_{v_0\in V}  \left(\displaystyle\sum_{\substack{w_0,\dots,w_{k-1} \\ \in W }} X_{\pi }\cdot \bisocond{ (\tilde{F},\tilde{P}),H}{\perm{w}{0}{k-1}}\right) \cdot \\ & \qquad \qquad \qquad I_{G,[k-1]} \cdot I_{G,\{0\}} \end{split}\\ 
                                     \begin{split} & = \displaystyle\sum_{v_0 \in V}  \left( \displaystyle\sum_{\substack{w_0,\dots,w_k \\ \in W}} X_{\pi'} \cdot\bisocond{ (\tilde{F},\tilde{P}),H}{\perm{w}{0}{k-1}}\right) \cdot \\ & \qquad \qquad \qquad I_{G,[k-1]} \cdot I_{G,\{0\}} \end{split}
\end{align*}
where $\pi= ((v_0,w_0),\dots, (v_{k-1}, w_{k-1})$ and $\pi'= ((v_0,w_0),\dots, (v_{k}, w_{k}))$. Here, we used the $\Liso[k+1](G,H)$ equation \ref{eq:ck2} to expand
\begin{align*}
X_{(v_0,w_0), \dots, (v_{k-1},w_{k-1})} = \displaystyle\sum_{w_k \in W} X_{(v_0,w_0), \dots, (v_{k-1},w_{k-1}), (v_{k},w_{k})} . 
\end{align*}
Since $X_{\pi'} \neq 0$ implies that $\{(v_0,w_0), \dots,(v_{k},w_{k})\}$ is a partial isomorphism, $X_{\pi'} \neq 0$ also implies that $I_{G,[k-1]}= I_{H,[k-1]}$ and $ I_{G,\{0\}}= I_{H,\{0\}}$. 
Therefore, we continue to rewrite
\begin{align*}
  \bisocond{(F,P),G}{\perm{v}{1}{k}} & = \displaystyle\sum_{\substack{w_0,\dots,w_k \\ \in W}} \displaystyle\sum_{v_0 \in V}  X_{\pi' } \cdot \bisocond{ (\tilde{F},\tilde{P}),H}{\perm{w}{0}{k-1}} \cdot I_{G,[k-1]} \cdot I_{G,\{0\}} \\ 
                                     & = \displaystyle\sum_{\substack{w_0,\dots,w_k \\ \in W}} \displaystyle\sum_{v_0\in V}  X_{\pi' } \cdot\bisocond{ (\tilde{F},\tilde{P}),H}{\perm{w}{0}{k-1}} \cdot I_{H,[k-1]} \cdot I_{H,\{0\}} \\ 
                                      \begin{split} & =\displaystyle\sum_{\substack{w_0,\dots,w_k \\ \in W}} \left(\displaystyle\sum_{v_0 \in V}  X_{\pi'} \right)\cdot\bisocond{ (\tilde{F},\tilde{P}),H}{\perm{w}{0}{k-1}} \cdot \\ &\qquad\qquad\qquad I_{H,[k-1]}  \cdot I_{H,\{0\}} \end{split}\\ 
                                     & = \displaystyle\sum_{\substack{w_0,\dots,w_k \\ \in W}}  X_{\pi''} \cdot \bisocond{ (\tilde{F},\tilde{P}),H}{\perm{w}{0}{k-1}} \cdot I_{H,[k-1]} \cdot I_{H,\{0\}} 
\end{align*}
where $\pi'' = \{(v_1,w_1),\dots,(v_k,w_k) \}$. Here, we used the $\Liso[k+1](G,H)$ equation \ref{eq:ck1} to collapse
\begin{align*}
\displaystyle\sum_{v_0 \in V} X_{(v_0,w_0), \dots, (v_{k-1},w_{k-1}), (v_{k},w_{k})} = X_{(v_1,w_1), \dots, (v_{k},w_{k})}. 
\end{align*}
Finally, we rewrite
\begin{align*}
\begin{split}
\bisocond{(F,P),G}{\perm{v}{1}{k}}  &= \displaystyle\sum_{\substack{w_1,\dots,w_k \\ \in W}}  \left(X_{\pi''} \right)\cdot\\
 & \left( \displaystyle\sum_{w_0 \in W} \bisocond{ (\tilde{F},\tilde{P}),H}{\perm{w}{0}{k-1}} \cdot I_{H,[k-1]} \cdot I_{H,\{0\}}\right) 
\end{split} \\
                                    & = \displaystyle\sum_{\substack{w_1,\dots,w_k \\ \in W}}  X_{\pi''} \cdot \left( \displaystyle\sum_{w_0 \in W} \bisocond{ (F,P),H}{\perm{w}{0}{k}}\right) \\ 
                                    & = \displaystyle\sum_{\substack{w_1,\dots,w_k \\ \in W}}  X_{\pi''} \cdot \left( \displaystyle\sum_{w_0 \in W} \bisocond{ (F,P),H}{\perm{w}{0}{k}}\right) \\ 
                                    & = \displaystyle\sum_{\substack{w_1,\dots,w_k \\ \in W}}  X_{\pi''} \cdot  \bisocond{ (F,P),H}{\perm{w}{1}{k}}  
\end{align*}
which finishes the proof of our lemma. 
\end{proof}

The proof of Theorem \ref{theo:4} is immediate from the following claim. 
\begin{claim}
  For all $k\geq 1$, and for all graphs $G$ and $H$, if
  $\Liso[k+1](G,H)$ has a real solution, then for all finite
  undirected graphs $F$ and an associated path decomposition $P$,
  $\biso((F,P),G)=\biso((F,P),H)$.
\end{claim}

\begin{proof}
The proof is by induction on $k$. The base case $k=0$ is trivial, since in this case, $\biso((F,P),G)$ and $\biso((F,P),H)$ merely count the 
number of vertices in graphs $G$ and $H$ respectively. For the inductive case, observe that if $\Liso[k+1](G,H)$ has a solution, then 
so does $\Liso[k'+1](G,H)$ for all $k'<k$. Hence, by inductive hypothesis, for all finite undirected graphs $F$ with a path decomposition $P$ of width at most $k'<k$, 
it holds that $\biso((F,P),G) = \biso((F,P),H)$. It remains to show that $\biso((F,P),G) = \biso((F,P),H)$ for all $F$ with a path decomposition $P$ of width $k$. 
Lemma \ref{lem:transfer} allows us to express
\begin{align*}
  \biso((F,P),G) &= \displaystyle\sum_{\substack{v_1,\dots,v_k \\ \in V}} \bisocond{(F,P),G}{\perm{v}{1}{k}} \\
                 & = \displaystyle\sum_{\substack{v_1,\dots,v_k \\ \in V}} \displaystyle\sum_{\substack{w_1,\dots,w_k \\ \in W }} X_{\pi''} \cdot \bisocond{(F,P),H}{\perm{w}{1}{k}} \\
                 & = \displaystyle\sum_{\substack{w_1,\dots,w_k \\ \in W }} \left( \displaystyle\sum_{\substack{v_1,\dots,v_k \\ \in V }}  X_{\pi''} \right)\cdot \bisocond{(F,P),H}{\perm{w}{1}{k}} \\
                 & = \displaystyle\sum_{\substack{w_1,\dots,w_k \\ \in W }} \bisocond{(F,P),H}{\perm{w}{1}{k}} \\
& = \biso((F,P),H)
\end{align*}
Using a similar argument to the case of treewidth-$k$ graphs in Appendix~\ref{app:treewidth}, it follows that for all finite 
 undirected graphs $F$ of pathwidth at most $k$, $\hom(F,G) = \hom(F,H)$ as well. This finishes the proof of Theorem \ref{theo:4}.
\end{proof}

\end{document}